\documentclass[11pt, a4paper]{amsart}
\usepackage{amsfonts}
\usepackage{amssymb}
\usepackage[dvips]{graphics}
\usepackage{epsfig}
\usepackage{euscript}
\usepackage{color}
\usepackage[pdftex]{hyperref}

 \newtheorem{thm}{Theorem}[section]
 \newtheorem{cor}[thm]{Corollary}
 \newtheorem{lemma}[thm]{Lemma}
 \newtheorem{prop}[thm]{Proposition}
 \theoremstyle{definition}
 \newtheorem{defn}[thm]{Definition}
 \newtheorem{rem}[thm]{Remark}

 \numberwithin{equation}{section}

\newcommand{\caA}{{\mathcal A}}
\newcommand{\caB}{{\mathcal B}}

\newcommand{\caH}{{\mathcal H}}

\newcommand{\caL}{{\mathcal L}}

\newcommand{\caN}{{\mathcal N}}
\newcommand{\caO}{{\mathcal O}}
\newcommand{\caP}{{\mathcal P}}

\newcommand{\bbE}{{\mathbb E}}

\newcommand{\bbI}{{\mathbb I}}

\newcommand{\bbN}{{\mathbb N}}

\newcommand{\bbR}{{\mathbb R}}

\newcommand{\bbZ}{{\mathbb Z}}

\newcommand{\id}{\bbI}

\newcommand{\ep}[1]{\mathrm{e}^{#1}}

\newcommand{\Inorm}{\vvvert}

\usepackage{mathtools}
\usepackage{etoolbox}

\DeclareFontFamily{U}{matha}{\hyphenchar\font45}
\DeclareFontShape{U}{matha}{m}{n}{
<-6> matha5 <6-7> matha6 <7-8> matha7
<8-9> matha8 <9-10> matha9
<10-12> matha10 <12-> matha12
}{}
\DeclareSymbolFont{matha}{U}{matha}{m}{n}

\DeclareFontFamily{U}{mathx}{\hyphenchar\font45}
\DeclareFontShape{U}{mathx}{m}{n}{
<-6> mathx5 <6-7> mathx6 <7-8> mathx7
<8-9> mathx8 <9-10> mathx9
<10-12> mathx10 <12-> mathx12
}{}
\DeclareSymbolFont{mathx}{U}{mathx}{m}{n}

\DeclareMathDelimiter{\vvvert} {0}{matha}{"7E}{mathx}{"17}%

\DeclarePairedDelimiterX{\normiii}[1]
{\vvvert}
{\vvvert}
{\ifblank{#1}{\:\cdot\:}{#1}}

\DeclareMathOperator\supp{supp}

\title[Product formulae]{Trotter product formulae for {\Large$*$}-automorphisms of quantum lattice systems}

\author{Sven Bachmann}
\address{Department of Mathematics \\ University of British Columbia \\ Vancouver, BC V6T 1Z2 \\ Canada}
\email{sbach@math.ubc.ca}

\author{Markus Lange}
\address{Mathematics Area \\ SISSA \\ Via Bonomea 265 \\ 34136, Trieste \\ Italy}
\email{mlange@sissa.it}

\date{\today}

\begin{document}

\begin{abstract}
We consider the dynamics $t\mapsto\tau_t$ of an infinite quantum lattice system that is generated by a local interaction. If the interaction decomposes into a finite number of terms that are themselves local interactions, we show that $\tau_t$ can be efficiently approximated by a product of $n$ automorphisms, each of them being an alternating product generated by the individual terms. For any integer $m$, we construct a product formula (in the spirit of Trotter) such that the approximation error scales as $n^{-m}$. Our bounds hold in norm, pointwise for algebra elements that are sufficiently well approximated by finite volume observables.
\end{abstract}

\maketitle

\section{Introduction}
For any two matrices $A$ and $B$, Lie proved the celebrated product formula
\begin{equation}\label{Lie}
\ep{A+B} = \lim_{n\to\infty}\left(\ep{A/n}\ep{B/n}\right)^n.
\end{equation}
There is a long line of similar formulae of increasing generality, pioneered by Trotter~\cite{trotter1959product}, simplified by Chernoff~\cite{chernoff1968note} for semigroups on Banach spaces, see e.g.~\cite{BratRob1}. In the particular setting of quantum mechanics where $A,B$ are densely defined semibounded self-adjoint operators and $\exp(i t A),\exp(i t B)$ and $\exp(i t (A+B))$ are the corresponding unitary groups, the product formula was proved under general assumptions by Kato~\cite{kato1978trotter} and Ichinose~\cite{ichinose1980product}, see also~\cite{ReedSimonI}. It plays a crucial role in functional integration, see in particular~\cite{nelson1964feynman}. For related results in the context of the quantum Zeno effect, we refer to~\cite{exner2005product}, and point further to~\cite{zagrebnov2019gibbs} for a historical overview, in particular in the case of Gibbs semigroups. 

The recent interest in proving general product formulae with explicit control of the rate of convergence has been motivated by two related developments in many-body quantum systems. On the one hand in quantum information theory, operator products arise as quantum circuits and a product formula is interpreted as a \emph{simulation algorithm} for the time evolution of a quantum system~\cite{Llo96,BerAhoCleSan06,wiebe2011simulating}. On the other hand in condensed matter physics, operator products are referred to as \emph{finite depth quantum circuits} and play a central role in the classification of gapped phases~\cite{chen2010local}, as they can be used to define the very notion of equivalence of states. In both cases, the concepts have recently been tested experimentally, see e.g.~\cite{NatureSimulationExp,SmiKimPolKno19}.

In both applications, the rate of convergence of the product formula to the full dynamics is of crucial importance: for quantum simulation because it determines the number of quantum gates required to simulate to a given error, for gapped phases because it relates to the degree of entanglement of ground states.  The standard general product formulae yield a rather poor scaling of either $n^{-1/2}$ or at best $n^{-1}$, see again~\cite{BratRob1}. In fact, in the case of Gibbs semigroups, there exist pairs of unbounded operators for which the norm difference is \emph{lower bounded} by $L(t) n^{-1}$, see~\cite{zagrebnov2019gibbs}. Furthermore, beyond the mere scaling, sharp constants are essential and may prove fatal in a many-body setting. Indeed, for a lattice system having $N$ degrees of freedom, the error diverges as $N\to\infty$, even in the strong operator topology, which is the natural topology as soon as $A,B$ are unbounded. This is related to the infrared catastrophe: If two states are locally close to each other but the error extends to spatial infinity, then they are in fact orthogonal.

In this work, we consider $d$-dimensional quantum lattice systems in the \emph{infinite volume limit}. The dynamics is an automorphism group $t\mapsto\tau^\Phi_t$ of the quasi-local algebra (which is a C*-algebra) generated by a local Hamiltonian formally given by
\begin{equation}\label{Hamiltonians}
H = \sum_{X}\Phi(X) = \sum_{j=1}^k K_j
\end{equation}
where the $K_j$'s correspond to an arbitrary grouping of the interaction terms. We provide product formulae and prove explicit bounds for sufficiently localized observables: For any $m\in\bbN$ there is a product automorphism denoted $\pi_{t,n}^{(m)}$ such that
\begin{equation}\label{InformalResult}
\left\Vert  \tau^\Phi_t(O) - \pi_{t,n}^{(m)}(O) \right\Vert
\leq C_{m,t,k}(O) n^{-m}
\end{equation}
for any almost local observable. The constant $C_{m,t,k}(O)$ depends on both the observable~$O$ and the Hamiltonian. The convergence we consider here is in the C*-norm, pointwise for sufficiently localized elements of the algebra. In other words, we consider the strong topology of operators 
acting on the C*-algebra. This is a purely Banach space result.

As in~(\ref{Lie}), $\pi_{t,n}^{(m)}$ are compositions of the individual dynamics generated by each $K_j$ individually. The general form of the product $\pi_{t,n}^{(m)}$ was proposed by Suzuki~\cite{suzuki1991general} although in the Hilbert space setting, see also~\cite{hatano2005finding}, and used recently by Childs et al.~\cite{PhysRevX.11.011020}, but only for finite systems with bounds that diverge in the volume.

From a technical point of view, we find it convenient to consider almost exclusively the \emph{interaction norms}, see Definiton~(\ref{def:int norm}) below, that measure the local size of extensive observables of the type~(\ref{Hamiltonians}). We extend their definition to be able to consider interactions localized in (possibly infinite) subsets of the lattice (see also~\cite{monaco2019adiabatic}) and remark that this construction is well-suited to discuss \emph{almost local observables}. Crucially, the unbounded $*$-derivation formally given by
\begin{equation*}
\sum_{X}i[\Phi(X),\cdot]
\end{equation*}
is well-defined on the set of almost local observables: it maps this set into itself, and we quantify explicitly the weakening of the localization induced by its action in terms of the interaction norms. 

We further wish to point out two related results. Firstly, a slightly different approach to a product formula was taken in~\cite{sahinoglu2021hamiltonian}, focussing on the `quasi-adiabatic' properties of product formulae, namely the error when projected onto the ground state space. Secondly, a similar spatial product factorization with sharp error bounds was derived in~\cite{haah2021quantum}: it is not based on the Trotter strategy but it uses rather directly the Lieb-Robinson bound, see also~\cite{bachmann2017local}.

While the results hold for a general decomposition~(\ref{Hamiltonians}), in applications the factors $K_j$ will be chosen so as to be commuting Hamiltonians, namely each $K_j$ is a sum of \emph{mutually commuting} interaction terms $\Phi_j(X)$, i.e. $[\Phi_j(X),\Phi_j(X')] = 0$ for all $X,X'$. Such Hamiltonians have the property that the corresponding automorphism $\tau^{\Phi_j}_t(O)$ is strictly local in that the support of the observable $O$ grows at most by the range of the interaction, uniformly in the time $t$. Propagation, which is obviously present in the full dynamics $\tau_t^{\Phi}(O)$ arises then through the alternating action of the automorphisms $\{\tau^{\Phi_j}:j=1,\ldots,k\}$. While the Lieb-Robinson bound is at the heart of the proofs, the product formulae provide a very clear picture of the mechanism of propagation. 

In the context of quantum simulation, much attention is given to the error made in the approximation upon truncation of the product formula to a finite number of terms. As we shall see, the error has a complicated dependence on a number of parameters and we shall discuss this in detail later. We already point out now (i) that the error is exponential
in the total time $t$, 
(ii) that the number of factors in~$\pi_{t,n}^{(m)}$ is proportional to $n$ and to the number $k$ of factors in the decomposition of the Hamiltonian, and that it is exponential in the order $m$ of approximation, 
and (iii) that unlike in the original Trotter product formula, the times involved in the various factors of~$\pi_{t,n}^{(m)}$ are not all equal, although they are all of order $\frac{t}{n}$; in fact, the time evolution runs backwards for a fraction of the factors, giving rise to a fractal path, see Figure~\ref{fig:Nice plots} at the end of Section~\ref{sec:ArbOrder}. 

Finally, we comment on the relation of the present work to~\cite{ALPUs}. While we start with a Hamiltonian evolution and approximate it with a finite depth quantum circuit, \cite{ALPUs} goes beyond, although only in one dimension. The starting point there is an almost locality preserving unitary (ALPU, an automorphism satisfying a Lieb-Robinson bound), which is not \emph{a priori} generated by a Hamiltonian. Generalizing the index defined in~\cite{GNVW}, and in the case that the index of the ALPU vanishes, Theorem~5.8 in~\cite{ALPUs} goes on to prove that the automorphism is in fact well approximated by a finite depth quantum circuit. The construction there is of a very different nature than ours and the successive layers of the circuit have increasing interaction range but decreasing strength and they act over a decreasing time interval.

\section{Setting} \label{sec:Setting}

Let $(\Gamma,\textbf{d})$ be a metric graph, where $\textbf{d}$ is the graph distance. We assume that $\Gamma$ is $d$-dimensional in the sense that $\sup_{x\in\Gamma}\vert \{y\in\Gamma: \textbf{d}(x,y) = r\}\vert =\omega (1+r)^{d-1}$. For any subset $X\subset\Gamma$, we define for any $r>0$ the set $X^{(r)} = \{x\in\Gamma:\textbf{d}(x,X)\leq r\}$ which is an $r$-fattening of the set $X$.

To each site $x \in \Gamma$ we associate a finite dimensional complex Hilbert space $\caH_x$ and define for any finite $\Lambda \subset \Gamma$, 
\begin{align*}
	\caH_\Lambda := \bigotimes_{x \in \Lambda} \caH_x \, \quad \textrm{ and } \quad \caA_\Lambda := \caB(\caH_\Lambda) \,,
\end{align*}
where $\caB(\caH_x)$ denotes the bounded linear operators over $\caH_x$. Moreover, we identify $A \in \caA_{\Lambda_0}$ with $A \otimes \id_{\Lambda\setminus\Lambda_0} \in \caA_\Lambda$ whenever $\Lambda_0\subset\Lambda$. 
With this we can inductively define the algebra of local observables 
\begin{align*}
	\caA_{\rm loc} := \bigcup_{\Lambda \in \caP_{\rm fin}(\Gamma)} \caA_{\Lambda}\,,
\end{align*}
where the union is taken over $\caP_{\rm fin}(\Gamma)$, the set of all finite subsets of $\Gamma$. If $O\in\caA_{\rm loc}$, then $\mathrm{supp}(O)$ is the smallest set $X$ such that $O\in\caA_X$. The completion of $\caA_{\rm loc}$ with respect to the norm topology is a $C^*$-algebra which is called the \textit{quasi-local algebra} and we denote it by $\caA$. The above construction of $\caA$ is completely standard and we refer to~\cite{BratRob2,BratRob1} for further details.

\subsection{Interactions and Hamiltonians}
\begin{defn}
An \textit{interaction} is a map $\Phi : \caP_{\rm fin}(\Gamma) \to \caA_{\rm loc}$ such that 
$$
\Phi(X)\in\caA_X,\qquad	\Phi(X) = \Phi(X)^*,
$$ 
for all $X \in \caP_{\rm fin}(\Gamma)$. 
\end{defn}

We turn the set of interactions into a Banach space in the following way.
Let $0 < p \leq 1$ and let
$$
	\xi_{b} : [0,\infty) \to (0,\infty)\,, \qquad \xi_{b}(x) = e^{-bx^p}\,,
$$ 
for any $b > 0$. The function $\xi_{b}$ is a decreasing, logarithmically superadditive function, namely $\xi_b(x+y)\geq \xi_b(x)\xi_b(y)$, that is summable in the following sense	
\begin{equation}\label{xi is integrable}
\|\xi_{b}\|_1 := \sup_{y \in \Gamma} \sum_{x \in \Gamma} \xi_{b}(\textbf{d}(x,y)) < \infty,
\end{equation}
since $\Gamma$ is finite dimensional.
\begin{defn}\label{def:int norm}
 Let $D(X):= \max\{\textbf{d}(x,y)\,:\, x, y \in X\}$ denote the diameter of the set $X\subset \Gamma$.  The \textit{interaction norm} of an interaction $\Phi$ is given by 
\begin{align}\label{eq:StandardNormInteraction}
	\Inorm\Phi\Inorm_b := \sup_{x \in \Gamma} \sum_{\substack{X \in \caP_{\rm fin}(\Gamma): \\ x \in X}} \frac{\|\Phi(X)\|}{\xi_{b}(D(X))}.
\end{align}
For fixed $b>0$ we denote the Banach space of interactions with finite $\Inorm\cdot\Inorm_b$-norm
by $\caB_b$ and set 
\begin{align*}
	\caB := \bigcup_{b >0} \caB_b\,.
\end{align*}
\end{defn}
An interaction $\Phi$ therefore belongs to $\caB$ if it belongs to at least one $\caB_b$. If $\Phi\in\caB_{b_0}$, it then follows by definition that $\Phi\in\caB_b$ for all $0<b\leq b_0$. Note that while each $\caB_b$ is a Banach space, their union $\caB$ is not.

Finally, we shall denote $\caB_\infty = \bigcap_{b>0} \caB_b$, namely $\Phi\in\caB_\infty$ if and only if $\Phi\in\caB_b$ for all $b>0$.

We point out the norm $\Inorm \cdot \Inorm_b$ indicates both the rate of decay of the interaction as well as its intensity in the sense that the total interaction at any given $x \in \Gamma$ is bounded by the interaction norm:
\begin{equation*}
	\sup_{x\in\Gamma}\Big\|\sum_{\substack{X \in \caP_{\rm fin}(\Gamma): \\ x \in X}} \Phi(X) \Big\|
	\leq \Inorm \Phi\Inorm_b\,.	
\end{equation*}
We will also need the notion of an interaction that is \textit{almost localized} in some set $Z \subset \Gamma$.
\begin{defn}
Let $Z \subset \Gamma$ and let $D_{Z}(X) := D(X) + \textbf{d}(X,Z)$. Let
\begin{align}\label{eq:RefinedNormInteraction}
	\Inorm\Phi\Inorm_{b,Z} := \sup_{x \in \Gamma} \sum_{\substack{X \in \caP_{\rm fin}(\Gamma): \\ x \in X}} \frac{\|\Phi(X)\|}{\xi_b(D_Z(X))}.
\end{align}
We denote $\caB_b(Z)$ the corresponding Banach space of interactions and define 
\begin{align*}
	\caB(Z) := \bigcup_{b > 0} \caB_{b}(Z) \,.
\end{align*} 
\end{defn}
Note that the replacement of the diameter by $D_{Z}(X)$ captures the decay of the interaction $\Phi(X)$ in the size of $X$ and in the distance from $X$ to $Z$. If $\Phi \in \caB_b(Z)$ then the total contribution of $\Phi$ at a given point $x$ is not only finite, but it decays with the distance of $x$ to $Z$,
\begin{align}\label{eq:TotalIntAtxFiniteAndDecays}
	\Big\|\sum_{\substack{X \in \caP_{\rm fin}(\Gamma): \\ x \in X}} \Phi(X) \Big\| \leq \Inorm \Phi\Inorm_{b,Z} \,\xi_b(\textbf{d}(x,Z))\,.
\end{align}
Indeed, we first write
\begin{equation}\label{usual}
\Big\Vert \sum_{\substack{X\in\caP_{\mathrm{fin}}(\Gamma):\\x\in X}}\Phi(X)\Big\Vert
\leq \sum_{\substack{X\in\caP_{\mathrm{fin}}(\Gamma):\\x\in X}} \frac{\Vert \Phi(X)\Vert }{\xi_b(D_Z(X))} \xi_b(D_Z(X)).
\end{equation}
Given $x\in X$, let $z,z_0\in Z$ and $x_0\in X$ be such that $\textbf{d}(x,Z) = \textbf{d}(x,z)$, $\textbf{d}(X,Z) = \textbf{d}(x_0,z_0)$. Then,
\begin{equation}\label{dxZ vx dXZ}
\textbf{d}(x,Z) = \textbf{d}(x,z) \leq \textbf{d}(x,z_0) \leq \textbf{d}(x,x_0) + \textbf{d}(x_0,z_0) \leq D_Z(X).
\end{equation}
Since $\xi_b$ is a decreasing function, we have that $\xi_b(D_Z(X)) \leq \xi_b(\textbf{d}(x,Z))$ which yields~(\ref{eq:TotalIntAtxFiniteAndDecays}) when plugged in~(\ref{usual}).

In general, the sum of an interaction is not convergent in $\caA$, but such a \emph{local Hamiltonian} defines a dynamics and a densely defined *-derivation on $\caA$. Moreover, if an interaction is almost localized in a finite set, then the sum is convergent and defines what we shall call an almost local observable in $\caA$. 

We start with the latter.
\begin{lemma}\label{lemma: finite Z interaction}
Let $Z\in\caP_{\mathrm{fin}}(\Gamma)$ and $\Phi\in\caB(Z)$. 
The sum
\begin{equation}\label{G of Phi}
G^\Phi := \sum_{X\in\caP_{\mathrm{fin}}(\Gamma)}\Phi(X)
\end{equation}
is convergent in $\caA$. Moreover,
\begin{equation*}
\Vert G^\Phi\Vert \leq C_b \Inorm \Phi \Inorm_{b,Z}\vert Z\vert 
\end{equation*}
for all $b>0$, where the right hand side is infinite whenever $\Phi\notin\caB_b(Z)$.
\end{lemma}
\begin{proof}
Let $b>0$ be such that $\Phi\in\caB_b(Z)$. Then for any $\Lambda\in \caP_{\mathrm{fin}}(\Gamma)$,
\begin{equation*}
\sum_{\substack{X\in\caP_{\mathrm{fin}}(\Gamma):\\X\cap\Lambda^c \neq \varnothing}}\Vert \Phi(X) \Vert \leq \sum_{x\in\Lambda^c}\sum_{X\ni x }\Vert \Phi(X) \Vert 
\end{equation*}
and we conclude by~(\ref{eq:TotalIntAtxFiniteAndDecays}) that
\begin{equation*}
\Big \Vert \sum_{\substack{X\in\caP_{\mathrm{fin}}(\Gamma):\\X\cap\Lambda^c \neq \varnothing}} \Phi(X) \Big \Vert\leq \Inorm \Phi \Inorm_{b,Z} \sum_{x\in\Lambda^c} \xi_b(\textbf{d}(x,Z)).
\end{equation*}
The integrability~(\ref{xi is integrable}) of $\xi_b$ and the finiteness of $Z$ imply that $\lim_{\Lambda\nearrow\Gamma}\sum_{x\in\Lambda^c} \xi_b(\textbf{d}(x,Z)) = 0$. The second claim follows from
\begin{equation*}
\Vert G^\Phi\Vert
\leq \sum_{x\in\Gamma}\sum_{\substack{X\in\caP_{\mathrm{fin}}(\Gamma):\\x\in X}}\frac{\Vert \Phi(X)\Vert}{\xi_b(D_Z(X))}\xi_b(D_Z(X))
\leq \Inorm \Phi \Inorm_{b,Z} \sum_{x\in\Gamma} \xi_b(\textbf{d}(x,Z))
\end{equation*}
where we used~(\ref{dxZ vx dXZ}). We decompose $\sum_{x\in\Gamma} = \sum_{n=0}^\infty\sum_{x:\textbf{d}(x,Z) = n}$ to finally get the bound
\begin{equation*}
\Vert G^\Phi\Vert
\leq\Inorm \Phi \Inorm_{b,Z}\sum_{n=0}^\infty \sum_{z\in Z}\sum_{x:\textbf{d}(x,z) = n}  \xi_b(n)
\leq C_b \Inorm \Phi \Inorm_{b,Z} \vert Z\vert
\end{equation*}
where $C_b = \omega\sum_{n=0}^\infty(1+n)^{d-1}\xi_b(n)$ is convergent.
\end{proof}
\begin{defn}\label{def:AlmostLocalObservable}
An \emph{almost local observable} is an element $O\in\caA$ for which there exists $Z\in\caP_{\mathrm{fin}}(\Gamma)$ and a $\Phi\in\caB(Z)$ such that $O = G^\Phi$. 
We denote the set of almost local observables by $\caL(Z)$, respectively $\caL_b(Z)$ whenever the rate $b$ is fixed.
\end{defn}

\noindent 
Slightly abusing language, we shall call $Z$ the almost support of $O\in \caL(Z)$. Moreover, we say that the interaction $\Phi$ in the definition is an interaction associated with $O$.

Let us now turn to interactions that are supported in the whole set $\Gamma$. They are locally finite, but the fact that they are extensive implies that a sum as in~(\ref{G of Phi}) is divergent. This suggests the following definition.
\begin{defn}\label{def:Hamiltonian}
A family of self-adjoint operators $H = \{H_\Lambda \, : \, \Lambda \in \caP_{\rm fin}(\Gamma)\}$ with $\supp(H_\Lambda) = \Lambda$ is a \textit{local Hamiltonian} if there exists an interaction $\Phi \in \caB$ such that 
\begin{align*}
	H_\Lambda = \sum_{X \subset \Lambda} \Phi(X)\,.
\end{align*}
We denote the set of local Hamiltonians by $\caL$.
\end{defn}
Let $H$ be a local Hamiltonian associated with an interaction $\Phi\in\caB_b$. Let $(\Lambda_n)_{n\in\bbN}$ be an increasing and absorbing sequence of finite sets. For an observable $O\in\caA_Z$, we have that if $n>m$
\begin{align*}
 \Vert [H_{\Lambda_n} - H_{\Lambda_m}, O] \Vert
&\leq  \sum_{\substack{X\subset\Lambda_n: \\ X\cap(\Lambda_m^c\cap Z)\neq\varnothing}}2\Vert O\Vert \Vert \Phi(X)\Vert \\
&\leq 2\Vert O\Vert \sum_{x\in Z}\sum_{\substack{X\ni x: \\ X\cap\Lambda_m^c\neq\varnothing}}
\frac{\Vert \Phi(X)\Vert}{\xi_b(D_Z(X))}\xi_b(D_Z(X)) \\
&\leq 2\Vert O\Vert \vert Z\vert\Inorm \Phi\Inorm_b \xi_b(\textbf{d}(Z,\Lambda_m^c))
\end{align*}
since $D_Z(X) = D(X) \geq \textbf{d}(Z,\Lambda_m^c)$. Hence $(i [H_{\Lambda_n} , O ])_{n\in\bbN}$ is a Cauchy sequence and $\lim_{n\to\infty}[H_{\Lambda_n},O]$ exists in $\caA$.

\subsection{Strongly continuous dynamics and derivations}
Let $H\in\caL$ with interaction $\Phi\in\caB$. The finite volume dynamics $\bbR\ni t\mapsto\ep{itH_\Lambda}O \ep{-itH_\Lambda}$ satisfies a Lieb-Robinson bound. 
While the proof runs along the general lines of~\cite{NacOgaSim06}, we reproduce it in the Appendix~\ref{app:LR-Bounds} in the specific setting of this paper; see also Section~4 in~\cite{NacSimYou19-1}. By standard arguments (see the previously cited reference or the original \cite{lieb1972finite}) the Lieb-Robinson bound implies that
\begin{align*}
	\tau_t^{\Phi}(O) = \lim_{\Lambda \nearrow \Gamma} \ep{it H_\Lambda} O \ep{-it H_\Lambda}	
\end{align*}
exists for all $O\in\caA_{\rm loc}$ and that it extents to a strongly continuous group of $^*$-automorphisms of $\caA$. 
The corresponding generator $\delta^\Phi$ of the dynamics $\tau_t^{\Phi}$ is given by 
\begin{align*}
	\frac{d}{dt} \tau_t^\Phi(O) = \tau_t^\Phi(\delta^\Phi(O))\,.
\end{align*}
A core for $\delta^\Phi$ is the local algebra $\caA_{\rm loc}$, see e.g.~Proposition~6.2.3 in~\cite{BratRob2} where $\delta^\Phi$ is explicitly given as the limit 
\begin{equation}\label{delta on Aloc}
\delta^\Phi(O) = \lim_{\Lambda\to\Gamma}i[H_\Lambda,O] = \sum_{X\in\caP_{\mathrm{fin}}(\Gamma)}i[\Phi(X),O],
\end{equation}
and the sum is convergent by the remark immediately after Definition~\ref{def:Hamiltonian}.

For a general Hamiltonian in $\caL$, neither $\tau_t^\Phi(O)$ nor $\delta^\Phi(O)$ is strictly local even if $O\in\caA_{\mathrm{loc}}$. However, we shall prove later that for any finite $Z$, $\caL(Z)$ belongs to the domain of $\delta^\Phi$ and it is invariant under the action of the derivation. Similarly, the Lieb-Robinson bound implies that $\caL(Z)$ is invariant under the action of~$\tau_t^\Phi$.

Let $r\in\bbN_0$. For any local observable $O \in \caA_Z$ we write
\begin{align}\label{eq:decompDynamics}
	\tau_t^\Phi(O) = \bbE_{Z^{(r)}}(\tau_t^\Phi(O)) + \sum_{n=r}^\infty \left(\bbE_{Z^{(n+1)}}(\tau_t^\Phi(O)) - \bbE_{Z^{(n)}}(\tau_t^\Phi(O))\right)
\end{align}
where  $\bbE_{X}$ is the projection onto the subalgebra $\caA_{X}$. Now, if $\Phi \in \caB_{b'}$, then the Lieb-Robinson bound for $\tau_t^\Phi$ implies that, for any $b''<b'$ and $n \in \bbN$,
\begin{align}\label{eq:AppOfLiebRobBound}
	\|\left(\bbE_{Z^{(n)}} - \textrm{id}\right)(\tau_t^\Phi(O))\|\leq  \frac{2 \|O\| |Z|}{M_{b'-b''}} \ep{\kappa(b',b'')|t|}\xi_{b''}(n)\,,
\end{align} 
where 
\begin{equation}\label{m}
M_\epsilon := \sup\{\vert X\vert\xi_{\epsilon}(D(X)):X\in\caP_{\mathrm{fin}}(\Gamma)\}
\end{equation}
is finite for any $\epsilon>0$, and the constant in the exponent is given by
\begin{equation}\label{DefKappa}
\kappa(b',b'') := 2\Inorm \Phi\Inorm_{b'} M_{b'-b''}.
\end{equation}

These estimates now yield the following proposition.
\begin{prop}\label{prop:DynamicsIsAlmostLocalObservable}
Let $Z\subset\Gamma$ and $O\in\caL_b(Z)$ with associated interaction~$\Psi$. Let $\Phi \in \caB_{b'}$ and let $\tau_t^\Phi$ be the corresponding dynamics. Then $\tau^\Phi_t(O)\in\caL_{b''}\left(Z\right)$ for any $b''<\min\{b,2^{-p} b'\}$. There is an interaction $\tau^\Phi_t(\Psi)$ associated with $\tau^\Phi_t(O)$ such that
\begin{equation*}
\Inorm \tau^\Phi_t(\Psi)\Inorm_{b'',Z} \leq C(b,b',b'') \ep{\kappa(b',\frac{1}{2}(b''+2^{-p}b')) \vert t\vert}  \Inorm \Psi\Inorm_{b,Z},
\end{equation*}
where the constant depends only on $b,b',b''$.
\end{prop}
\begin{proof}
Let $\Psi\in\caB_{b}(Z)$ be the interaction associated with $O$, namely $O = \sum_{X\in\caP_{\mathrm{fin}}(\Gamma)}\Psi(X)$. We construct an interaction, denoted $\tau_t^\Phi(\Psi)$, such that $\tau_t^\Phi(O) = \sum_{X\in\caP_{\mathrm{fin}}(\Gamma)}\tau_t^\Phi(\Psi)(X)$ as follows. We decompose each $\tau_t^\Phi(\Psi(Y))$ according to~(\ref{eq:decompDynamics}) with $r=0$ and gather contributions supported in a given set $X$ to get
\begin{equation}\label{LocalDecomp}
\tau_t^\Phi(\Psi)(X) := \bbE_{X}(\tau_t^\Phi(\Psi(X))) +
\sum_{n=1}^\infty\sum_{\substack{Y\in\caP_{\mathrm{fin}}(\Gamma): \\ X=Y^{(n)}}}\left(\bbE_{Y^{(n)}}-\bbE_{Y^{(n-1)}}\right)(\tau_t^\Phi(\Psi(Y))).
\end{equation}
Using~(\ref{eq:AppOfLiebRobBound}), all terms in the sum are bounded above by
\begin{equation}\label{E-E}
\left\Vert (\bbE_{Y^{(n)}}-\bbE_{Y^{(n-1)}})(\tau_t^\Phi(\Psi(Y)))\right\Vert
\leq \frac{2\Vert \Psi(Y)\Vert \vert Y^{(n)}\vert}{M_{b'-\tilde b}}\ep{ \kappa(b',\tilde b) |t|}\xi_{\tilde b}(n),
\end{equation}
for any $\tilde b<b'$. The first term is bounded above by $\Vert \Psi(X)\Vert$ since $\bbE_{X}$ is a projection; for simplicity, we shall rather use the estimate~(\ref{E-E}) with $n=0$ in the following. 

Let $x\in\Gamma$. We claim that $\sum_{\substack{X \in \caP_{\rm fin}(\Gamma) : \\ x\in X}}\frac{\Vert \tau_t^\Phi(\Psi)(X)\Vert }{\xi_{b''}(D_{Z}(X))}$ is uniformly bounded. We shall use~(\ref{LocalDecomp}) together with the estimate~(\ref{E-E}). If $x\in X= Y^{(n)}$, then in particular $B_n(x)\cap Y\neq\varnothing$ and so
\begin{multline*}
\sum_{\substack{X \in \caP_{\rm fin}(\Gamma) : \\ x\in X}} \frac{\Vert \tau_t^\Phi(\Psi)(X)\Vert }{\xi_{b''}(D_{Z}(X))} \\
\leq \frac{2\omega}{M_{b'-\tilde b}}\sum_{n=0}^\infty \sum_{y\in B_{n}(x)}\sum_{\substack{Y \in \caP_{\rm fin}(\Gamma) : \\ y\in Y}}
\frac{ \Vert \Psi(Y)\Vert}{\xi_{b}(D_Z(Y))}\vert Y\vert \xi_{b-b''}(D_Z(Y))\frac{\xi_{\tilde b}(n)}{\xi_{b''}(2n)}(1+n)^d\ep{\kappa(b',\tilde b)\vert t\vert}.
\end{multline*}
In this estimate, we firstly recalled that $X = Y^{(n)}$ and used $D_Z(Y^{(n)})\leq D_Z(Y) + 2n$ to conclude that $\xi_{b''}(D_{Z}(X))\geq \xi_{b''}(D_{Z}(Y))\xi_{b''}(2n)$, secondly factorized $\frac{1}{\xi_{b''}} = \frac{\xi_{b-b''}}{\xi_{b}}$. We also used that $\vert Y^{(n)}\vert\leq \omega\vert Y\vert (1+n)^d$. Hence,
\begin{equation*}
\sup_{x\in\Gamma}\sum_{\substack{X \in \caP_{\rm fin}(\Gamma) : \\ x\in X}}\frac{\Vert \tau_t^\Phi(\Psi)(X)\Vert }{\xi_{b''}(D_{Z}(X))}\leq S(\tilde b,b'')\frac{2\omega M_{b-b''}}{M_{b'-\tilde b}}\Inorm\Psi\Inorm_{b,Z}\ep{\kappa(b',\tilde b)\vert t\vert}
\end{equation*}
where we used that $S(\tilde b,b'') = \sum_{n=0}^\infty\frac{\xi_{\tilde b}(n)}{\xi_{b''}(2n)}(1+n)^{2d}$  is finite because we can pick $\tilde b$ such that $b''<2^{-p}\tilde b$ since $b''<2^{-p}b'$ to ensure the convergence of the series. For simplicity, we make the specific choice $\tilde b = \frac{1}{2}(b'' + 2^{-p}b')$ and let $C(b,b',b'') = 2\omega S(\tilde b,b'') \frac{ M_{b-b''}}{M_{b'-\tilde b}}$.
\end{proof}
\begin{rem}
(i) A less detailed but clearer way to formulate the result would be that the *-subalgebra of almost local observables supported in $Y$ is an invariant subspace for $\tau_t^\Phi$ for $t$ in a compact interval.\\
(ii) One could wish to take the propagation into account in this result and prove rather that $\tau_t^\Phi(\caL(Z)) \subset \caL(Z^{(v\vert t\vert)})$, at least in the case of an interaction $\Phi$ that decays exponentially. This is of course true as $\caL(Z^{(v\vert t\vert)})$ and $\caL(Z)$ are equal as sets, but equipped with different norms. Since however the bound would still be superpolynomially large in time (because the support of each individual interaction term grows with time and hence the decay rate of the interaction does worsen), there is no real gain in doing so.
\end{rem}

The derivation $\delta^\Phi$ associated with a local interaction is in general unbounded on $\caA$ and accordingly not everywhere defined. As pointed out earlier, $\caA_{\mathrm{loc}}$ is a core on which it is given explicitly as the limit of a commutator. We prove that $\delta^\Phi$ extends to the set of almost local observables and that, as for the automorphism $\tau^\Phi_t$, the sets $\caL(Y)$ are invariant under the action of $\delta^\Phi$. 

Instead of considering $\delta^\Phi$ as an unbounded operator on observables, we find it more convenient to define it on the set of interactions $\caB$ and to show that it extends to a bounded linear operator $\caB_b\to\caB_{b'}$ for appropriate pairs $(b,b')$. A similar approach was in fact already taken in~\cite{NacSimYou19-1}.
\begin{defn}\label{def:delta of psi}
Let $\Phi,\Psi\in\caB$. The interaction $\delta^\Phi(\Psi)$ is defined by
\begin{equation}\label{eq:DerivationForTwoInteractions}
	\delta^\Phi(\Psi)(X) := \sum_{\substack{Y,Y' \in \caP_{\rm fin}(\Gamma): \\ Y \cap Y' \ne \varnothing, \, Y\cup Y' = X}} i [\Phi(Y),\Psi(Y')]
\end{equation}
for any $X\in \caP_{\rm fin}(\Gamma)$. 
\end{defn}
\noindent Note that the condition $Y \cap Y' \ne \varnothing$ is only for clarity since the commutator vanishes if it is not satisfied.
\begin{rem}
If $O\in\caA_Z$ and $\Psi$ is the interaction trivially associated with it, namely $\Psi(Z) = O$  and $\Psi(X) = 0$ otherwise, then the definition above yields an interaction such that 
\begin{equation*}
\sum_{X\in \caP_{\rm fin}(\Gamma)}\delta^\Phi(\Psi)(X)
=\sum_{Y\in \caP_{\rm fin}(\Gamma)} i [\Phi(Y),O] = \delta^\Phi(O)
\end{equation*}
as in~(\ref{delta on Aloc}), justifying the notation $\delta^\Phi$. 
\end{rem}
\begin{prop}\label{prop:ActionOfDeltas}
Let $Z\subset \Gamma$ and let $\Psi\in\caB_b(Z)$. Let $\Phi \in \caB_{b'}$. If $\delta^\Phi(\Psi)$ is defined as in~(\ref{eq:DerivationForTwoInteractions}), then $\delta^\Phi(\Psi)\in\caB_{b''}(Z)$ for any $b''<\min\{b,b'\}$ and 
\begin{equation}\label{eq:InteractionNormOfDerviationDelta}
\Inorm \delta^\Phi(\Psi)\Inorm_{b'',Z}\leq 4 M_{\min\{b,b'\}-b''}\Inorm \Phi\Inorm_{b'} \: \Inorm  \Psi \Inorm_{b,Z}.
\end{equation}
\end{prop}
\begin{proof}
For $x\in\Gamma$, we wish to estimate
\begin{equation}\label{eq:FirstEstimateDerivationNorm}
	\sum_{\substack{X \in \caP_{\rm fin}(\Gamma): \\ x \in X}} 
	 \sum_{\substack{Y,Y' \in \caP_{\rm fin}(\Gamma): \\ Y \cap Y' \ne \varnothing, \, Y\cup Y' = X}} \frac{2\|\Phi(Y)\| \|\Psi(Y')\|}{\xi_{b''}(D_Z(X))}.
\end{equation}
There are two possibilities for the second sum, either $x\in Y$  or $x\in Y'\setminus Y$. In the first case we can bound the sum by
\begin{align*}
\sum_{\substack{Y \in \caP_{\rm fin}(\Gamma): \\ x \in Y}} \frac{ 2 \|\Phi(Y)\|}{\xi_{b'}(D(Y))}& \sum_{y \in Y} \sum_{\substack{Y' \in \caP_{\rm fin}(\Gamma): \\ y \in Y'}} \frac{\|\Psi(Y')\|}{\xi_{b}(D_Z(Y'))} \frac{\xi_{b'}(D(Y)) \xi_{b}(D_Z(Y'))}{\xi_{b''}(D_Z(Y\cup Y'))}
\end{align*}
and in the second case the bound is similar. Now $\textbf{d}(Z,Y\cup Y')\leq \textbf{d}(Z,Y')$. What is more, since $Y,Y'$ are not disjoint, $D(Y\cup Y') \leq D(Y) + D(Y')$ so that monotonicity and superadditivity yield
\begin{align*}
\frac{\xi_{b'}(D(Y)) \xi_{b}(D_Z(Y'))}{\xi_{b''}(D_Z(Y\cup Y'))}
&\leq \frac{\xi_{\min\{b,b'\}}(D(Y) + D(Y') + \textbf{d}(Z,Y'))}{\xi_{b''}(D(Y) + D(Y') + \textbf{d}(Z,Y\cup Y'))}\\
&\leq \xi_{\min\{b,b'\}-b''}(D(Y) + D_Z(Y'))
\end{align*}
Since $b''<\min\{b,b'\}$, we conclude that
\begin{align*}
		\Inorm\delta^{\Phi}(\Psi)\Inorm_{b'',Z} 
		\leq C \Inorm \Phi\Inorm_{b'} \Inorm \Psi\Inorm_{b,Z}\,,
\end{align*}
where $C = 4\sup_{Y,Y' \in \caP_{\rm fin}(\Gamma) } \left\{|Y|  \xi_{\min\{b,b'\}-b''}(D(Y)+ D_Z(Y'))\right\}\leq 4M_{\min\{b,b'\}-b''}$, as announced.
\end{proof}
\begin{rem}
We note that this is valid for any set $Z$, not necessarily finite. If $Z\in \caP_{\rm fin}(\Gamma)$, then by the proposition both $\Psi$ and $\delta^\Phi(\Psi)$ correspond to almost local observables $G^\Psi$ and $G^{\delta^\Phi(\Psi)}$ in $\caL(Z)$ and the map $G^\Psi \mapsto \delta^\Phi(G^\Psi):= G^{\delta^\Phi(\Psi)}$ provides the announced extension of $\delta^\Phi$ from $\caA_{\mathrm{loc}}$ to the set of almost local observables in $\caA$. Moreover, the proposition shows that if $\Phi\in\caB_{b'}$, then the  map $\delta^\Phi$ is well-defined for any interaction in $\caB$ and that it is a bounded linear operator $\caB_b(Z)\to\caB_{b''}(Z)$ for any $b''<\min\{b,b'\}$ and any subset $Z$. The upper bound on $\Vert \delta^\Phi\Vert_{\caL(\caB_{b}(Z),\caB_{b''}(Z))}$ provided by the proof diverges as $b''\to\min\{b,b'\}$, but it can be taken to be uniform in $Z$.
\end{rem}

We conclude this section with a joint corollary of Proposition~\ref{prop:DynamicsIsAlmostLocalObservable} \& \ref{prop:ActionOfDeltas}. For any $Z \in \caP_{\rm fin}(\Gamma)$, the set $\caL(Z)$ of almost local observables is invariant under the action of $\tau_t^\Phi$ and $\delta^\Phi$ for any $t$ in a compact interval. It follows in particular that Duhamel's formula and its iterates to arbitrary order are well-defined.

\begin{cor}
Let $H\in\caL$ with interaction $\Phi\in\caB$.
The function $t\mapsto \tau_t^\Phi$ is infinitely often strongly differentiable on the algebra of almost local observables. 
In particular if $O\in\caL(Z)$ for some $Z\in\caP_{\mathrm{fin}}(\Gamma)$, then Duhamel's formula 
\begin{equation}
\tau^\Phi_t(O) 
= O + \sum_{j=1}^{n-1} \frac{t^j}{j!}\left(\delta^\Phi\right)^j(O)
+ \int_{\Sigma^n_t}\tau^\Phi_{s_n}\left((\delta^\Phi)^n(O)\right)d^n\!s \label{Duahmel}
\end{equation}
is well-defined for any $n\in\bbN$. We denoted $\Sigma^n_t :=\{0\leq s_1\leq \ldots \leq s_n\leq t\}$ and $d^n\!s = ds_n\ldots ds_1$.
\end{cor}

\section{A product automorphism of lowest order} \label{Sec:Anchor}

With these preliminaries at hand, we now prove the validity of what is sometimes referred to as the symmetric Trotter product formula in the context of an infinite quantum lattice system. Let $H\in\caL$ be a local Hamiltonian with interaction $\Phi\in\caB$. We assume that
\begin{equation}\label{K's}
H_\Lambda = \sum_{j=1}^k K_{j,\Lambda}
\end{equation}
where $K_{j,\Lambda}\in\caL$ are local Hamiltonians with corresponding interactions $\Phi_{j}\in\caB$. 
We denote $\tau^j_t = \tau_t^{\Phi_j}$ and $\delta^j = \delta^{\Phi_j}$. Let us first consider the automorphism of $\caA$ defined by
\begin{equation}\label{Lowest order symmetric}
\sigma^{(1)}_t(O) := \tau^1_{t/2}\circ\cdots\circ\tau^k_{t/2}\circ\tau^k_{t/2}\circ\cdots\circ\tau^1_{t/2}(O).
\end{equation}
We assume that $\Phi\in \caB_{b'}$ and $\Phi_j\in\caB_{b_j}$ for $j=1,...,k$. We denote
\begin{equation}\label{MaxNorm}
\caN := \max\left\{\Inorm \Phi \Inorm_{b'}, \Inorm \Phi_1 \Inorm_{b_1},\ldots, \Inorm \Phi_k \Inorm_{b_k}\right\}.
\end{equation}
\begin{thm}\label{Thm: anchor}
Let $n \in \bbN$, $t \in \bbR_+$, $\mu = \frac{t}{n}$ and let 
\begin{equation*}
\pi^{(1)}_{t,n}(O) := \left(\sigma^{(1)}_\mu\right)^n(O).
\end{equation*}
Let $Z\in\caP_{\mathrm{fin}}(\Gamma)$ and let $b>0$. There are positive constants $C,v$ depending only on $b,b',b_1,\ldots,b_k$ and $k$ such that if $O \in \caL_b(Z)$,
\begin{equation*}
\left\Vert  \tau^\Phi_t(O) - \pi_{t,n}^{(1)}(O) \right\Vert
\leq C \Inorm \Psi \Inorm_{b,Z}\vert Z\vert  \frac{(\caN t)^3 \ep{v t}}{n^2}.
\end{equation*}
Here, $\Psi$ is an interaction associated with $O$.
\end{thm}
Note that in the case $k=2$, namely $H = A+B$, and in finite volume, the product automorphism reduces to the adjoint action of $\left(\ep{i\frac{t}{2n}B}\ep{i\frac{t}{n}A}\ep{i\frac{t}{2n}B}\right)^n$, which is indeed well-known to converge to the adjoint action of $\ep{i t (A+B)}$ as $n\to\infty$.
While the convergence is trivially uniform in the observable $O$ in finite volume (the finite volume algebras being finite dimensional), this uniformity cannot be expected to hold in the infinite volume limit. Pointwise convergence in norm is a consequence of the general Banach space theory originally due to Chernoff, see again~\cite{BratRob1}. In this context, the interest of Theorem~\ref{Thm: anchor} is that it provides an  explicit rate of convergence $n^{-2}$, for any almost local $O\in\caL(Z)$ and any finite set $Z$ (a fortiori for any strictly local observable).

\begin{proof}[Proof of Theorem~\ref{Thm: anchor}]
We first decompose the time interval $[0,t]$ in $n$ subintervals of width $\mu = \frac{t}{n}$ to get the following telescopic sum:
\begin{equation}\label{Telescopic}
 \tau^\Phi_t(O) - \pi_{t,n}^{(1)}(O)
= \sum _{j = 0}^{n-1} \left(\sigma^{(1)}_\mu\right)^j\left((\tau^\Phi_\mu -\sigma^{(1)}_\mu)\left( (\tau^\Phi_{\mu})^{n-j-1}(O) \right)\right).
\end{equation}
For any almost local observable $\widetilde O\in \caL_{b''}(Z)$, we see that
\begin{equation*}
\frac{d}{ds}\left.\left(\sigma^{(1)}_{s}\circ\tau^\Phi_{-s}\right)(\widetilde O)\right\vert_{s=0}
= (2\sum_{j=1}^k \frac{\delta^j}{2} - \delta^\Phi)(\widetilde O) = 0
\end{equation*}
by~(\ref{Lowest order symmetric}) and~(\ref{K's}). Similarly, but with a little more algebra, 
\begin{align*}
\frac{d^2}{ds^2}\left.\left(\sigma^{(1)}_{s}\circ\tau^\Phi_{-s}\right)(\widetilde O)\right\vert_{s=0}
&= \bigg(\frac{1}{4}\sum_{j=1}^k \Big\{\sum_{l=1}^j\delta^l \delta^j + \sum_{l=j+1}^k\delta^j \delta^l + \delta^j \sum_{l=1}^k \delta^l-2\delta^j\delta^\Phi \\
&\qquad\qquad\quad+ \sum_{l=1}^k \delta^l \delta^j + \sum_{l=j+1}^k\delta^l \delta^j + \sum_{l=1}^j\delta^j \delta^l  -2\delta^j\delta^\Phi\Big\} \\
&\qquad-2\sum_{j=1}^k \frac{1}{2}\delta^j\delta^\Phi
+\delta^\Phi\delta^\Phi
\bigg)(\widetilde O ).
\end{align*}
Writing $\sum_{j=1}^k\sum_{l=1}^j \delta^l \delta^j = (\delta^\Phi)^2 - \sum_{j=1}^k\sum_{l=j+1}^k \delta^l \delta^j$ and proceeding similarly for the second-to-last term of the second line, we conclude that this derivative vanishes again by $\delta^\Phi = \sum_{j=1}^k\delta^j$. Thus,
\begin{align}%
\tau^\Phi_\mu(\widetilde O) -\sigma^{(1)}_\mu(\widetilde O)
&=-\left.\left(\sigma^{(1)}_{s}\circ\tau^\Phi_{-s}\right)(\tau^\Phi_\mu(\widetilde O))\right\vert_{s=0}^{s=\mu} \nonumber \\
&= - \int_{\Sigma_\mu^3} \frac{d^3}{ds_3^3}\left(\sigma^{(1)}_{s_3}\circ\tau^\Phi_{-s_3}\right)(\tau^\Phi_\mu(\widetilde O)) d^3s.\label{Infinitesimal steps}
\end{align}

Distributing the three derivatives across the $2k$ factors of $\sigma^{(1)}_{s_3}\circ\tau^\Phi_{-s_3}$ 
inserts three derivations to the product $\tau^1_{s_3/2}\circ\cdots\circ\tau^k_{s_3}\circ\cdots\circ\tau^1_{s_3/2}\circ\tau^\Phi_{-s_3}(\tau^\Phi_\mu(\widetilde O))$. By Propositions~\ref{prop:DynamicsIsAlmostLocalObservable} and~\ref{prop:ActionOfDeltas} all terms are well-defined and belong to $\caL(Z)$ since $\tau^\Phi_\mu(\widetilde O) \in\caL(Z)$. Specifically, each application of an automorphism yields an exponential factor and the three derivations provide an additional $\caN^3$, see~(\ref{MaxNorm}).
Moreover, each of these operations yields an additional multiplicative constant, resulting in an overall factor that depends on $k$ and on the rates $b',b_1,\ldots,b_k$, but it is independent of $n$. It follows that 
for any $\tilde b < \min\{b'', 2^{-p}b', 2^{-p}b_1, \ldots, 2^{-p}b_k\}$,
 the interaction norm of each term is bounded by $C\caN^3 \ep{c((k+1)s_3 + \mu)}\Inorm\widetilde O\Inorm_{b'',Z}$, where $(k+1)s_3 + \mu$ is the total time (in absolute value) involved in~$\left(\sigma^{(1)}_{s_3}\circ\tau^\Phi_{-s_3}\right)(\tau^\Phi_\mu(\widetilde O))$ and the constant $c$ is the maximum of all $(2k+1)$ constants $\kappa(\cdot,\cdot)$ given by Proposition~\ref{prop:DynamicsIsAlmostLocalObservable}. Finally, we recall from~\eqref{Telescopic} that $\widetilde O = (\tau^\Phi_{\mu})^{n-j-1}(O) = \tau^\Phi_{(n-j-1)\mu}(O)$ with $O\in\caL_b(Z)$, so that its $b''$-interaction norm (with $b''<\min\{b,2^{-p} b'\}$) is bounded by $ C\ep{c(n-j-1)\mu}\Inorm \Psi \Inorm_{b,Z}$, where $C,c$ are, again, independent of $n$. Gathering all estimates, 
\begin{align*}
\left\Inorm \left(\tau^\Phi_\mu -\sigma^{(1)}_\mu\right)\left( (\tau^\Phi_{\mu})^{n-j-1}(O) \right) \right\Inorm_{\tilde b,Z}
&\leq C \Inorm \Psi \Inorm_{b,Z} \caN^3 \ep{c (n-j)\mu}\int_{\Sigma_\mu^3} \ep{c(k+1)s_3} d^3s \\
&\leq C \Inorm \Psi \Inorm_{b,Z} \ep{c (n-j)\mu}\frac{(\caN\mu)^3}{3!}\ep{c(k+1)\mu}.
\end{align*}
Since $\sigma^{(1)}_\mu$ preserves the operator norm, each term of~(\ref{Telescopic}) is bounded by
\begin{equation*}
\left\Vert (\sigma^{(1)}_\mu)^j((\tau^\Phi_\mu -\sigma^{(1)}_\mu) ((\tau^\Phi_{\mu})^{n-j-1}(O) ))\right\Vert
\leq C \Inorm \Psi \Inorm_{b,Z} \vert Z\vert \ep{v t }(\caN\mu)^3
\end{equation*}
by Lemma~\ref{lemma: finite Z interaction}, where we used that $c(n-j+(k+1))\mu \leq vt$, where $v=c(k+2)$. This estimate being uniform across the $n$ terms of~(\ref{Telescopic}), we immediately conclude with the claim of the theorem.
\end{proof}

\section{Arbitrary order}\label{sec:ArbOrder}

The symmetric Trotter formula discussed in the previous section has an error of order $n^{-2}$. As pioneered by Suzuki in, e.g.~\cite{suzuki1991general}, a recursive construction can be build upon it to generate higher order product formulae. We now show that they too extend to the infinite volume setting. 

\subsection{Time reversal}
Let us recall the automorphism $\sigma_t^{(1)}$ defined for all $t\in\bbR$ by~(\ref{Lowest order symmetric}). Since $\sigma^{(1)}_t$ is a composition of automorphisms, it is an automorphism, but the fact that the individual factors do not commute with each other breaks the group property of the fundamental time evolution $\tau^\Phi_t\circ\tau^\Phi_{s} = \tau^\Phi_{t+s}$. However, the specific `symmetric' form of~(\ref{Lowest order symmetric}) implies that
\begin{equation*}
\sigma^{(1)}_{-t}\circ\sigma^{(1)}_t = \mathrm{id}.
\end{equation*}
A product automorphism having this property shall be called time-reversal symmetric. Theorem~\ref{Thm: anchor} shows that, despite its label $(1)$, the corresponding product automorphism $\pi_{t,n}^{(1)}$ is in fact a second order approximation of $\tau_t^\Phi$. This improvement from any odd order to the next even order is in fact general for time-reversal symmetric product approximations.
\begin{prop}\label{prop: TRI}
Let $m\in\bbN$ and let $\{\sigma^{(2m-1)}_\mu:\mu\in\bbR\}$ be an $(2m-1)$-th order product approximation of $\tau^\Phi_\mu$ in the sense that
\begin{equation}\label{lower orders}
\left.\frac{d^j}{d\mu^j} \left(\tau^\Phi_\mu(\widetilde O) - \sigma^{(2m-1)}_\mu(\widetilde O)\right)\right\vert_{\mu = 0} = 0\qquad(j\in\{0,\ldots,2m-1\})
\end{equation}
for any $\widetilde O \in \caL(Y)$. If $\sigma^{(2m-1)}_\mu$ is time-reversal symmetric,
\begin{equation*}
\sigma^{(2m-1)}_{-\mu}\circ\sigma^{(2m-1)}_\mu = \mathrm{id},
\end{equation*}
then it is a $(2m)$-th order approximation of $\tau^\Phi_\mu$.
\end{prop}
\begin{proof}
The identity
\begin{equation*}
\widetilde O = \tau^\Phi_{-\mu}\circ(\tau^\Phi_\mu-\sigma^{(2m-1)}_\mu)(\widetilde O) + (\tau^\Phi_{-\mu}-\sigma^{(2m-1)}_{-\mu})\circ\sigma^{(2m-1)}_\mu(\widetilde O) + \sigma^{(2m-1)}_{-\mu}\circ\sigma^{(2m-1)}_\mu(\widetilde O)
\end{equation*}
and time-reversal symmetry imply that
\begin{equation*}
\tau^\Phi_{-\mu}\circ(\tau^\Phi_\mu-\sigma^{(2m-1)}_\mu)(\widetilde O) + (\tau^\Phi_{-\mu}-\sigma^{(2m-1)}_{-\mu})\circ\sigma^{(2m-1)}_\mu(\widetilde O) = 0.
\end{equation*}
The derivative of order $2m$ of this equation at $\mu=0$ reduces by~(\ref{lower orders}) to
\begin{equation*}
\left.\frac{d^{2m}}{d\mu^{2m}}\left(\tau_\mu^\Phi-\sigma^{(2m-1)}_\mu\right)(\widetilde O)\right\vert_{\mu=0}
+ \left.\frac{d^{2m}}{d\mu^{2m}}\left(\tau_{-\mu}^\Phi-\sigma^{(2m-1)}_{-\mu}\right)(\widetilde O)\right\vert_{\mu=0} =0,
\end{equation*}
which concludes the proof since the two derivatives of even order are equal.
\end{proof}

\subsection{Suzuki's Ansatz} 

We now recall Suzuki's inductive construction~\cite{suzuki1991general} of higher order product formulae, translated in the present language of automorphisms. Since Section~\ref{Sec:Anchor} provided a time reversal symmetric approximation of order $2$, we shall use it to anchor the induction. For that, we first let $\sigma_\mu^{(2)}:=\sigma_\mu^{(1)}$ for any $\mu\in\bbR$.

Let $\sigma_\mu^{(2m)}$ be a time-reversal symmetric $(2m)$-th order product approximation of $\tau_\mu^\Phi$ in the sense of~(\ref{lower orders}). A higher order approximation can be constructed as follows. Let $r = 2\ell + 1 \geq 3$ be an odd integer and let
\begin{equation}\label{s coefficients}
s_{m} := \frac{1}{(r-1)-(r-1)^{\frac{1}{2m+1}}}.
\end{equation}
We immediately point out firstly that $2\ell s_m + (1-(r-1)s_m) = 1$ and secondly that $ -1+(r-1)s_m =s_m(r-1)^{\frac{1}{2m+1}}$, and so
\begin{equation}\label{scalings}
(r-1)s_m^{2m+1} + \left(1-(r-1)s_{m}\right)^{2m+1} = 0.
\end{equation}
We now define the following product automorphisms:
\begin{equation}\label{Higher order construction}
\sigma_\mu^{(2m+1)} :=\left(\sigma_{s_{m}\mu}^{(2m)}\right)^\ell\circ\sigma_{(1-(r-1)s_{m})\mu}^{(2m)}\circ\left(\sigma_{s_{m}\mu}^{(2m)}\right)^\ell,
\quad \sigma_\mu^{(2m+2)}:=\sigma_\mu^{(2m+1)}.
\end{equation}
This procedure provides, given an odd integer $r$, a family of automorphisms $\{\sigma_\mu^{(2m+1)}:m\in\bbN\}$ parametrized by $\mu\in\bbR$.

For the following result, recall the setting of Section~\ref{Sec:Anchor}.
\begin{thm}\label{Thm: TRI-Trotter}
Let $\ell \geq 1$ and $r = 2\ell +1$. For all $m\geq 1$, $\{\sigma_{s}^{(m)}:s\in\bbR\}$ is time-reversal symmetric. Let $n \in \bbN$, $t \in \bbR_+$, and $\mu = \frac{t}{n}$. Define
\begin{equation*}
\pi^{(m)}_{t,n} := \left(\sigma^{(m)}_\mu\right)^n.
\end{equation*}
Let $Z\in\caP_{\mathrm{fin}}(\Gamma)$ and let $b>0$. There are positive constants $C,v$ such that if $O \in \caL_b(Z)$, then 
\begin{equation}\label{main formula}
\left\Vert  \tau^\Phi_t(O) - \pi_{t,n}^{(m)}(O) \right\Vert
\leq C \Inorm\Psi\Inorm_{b,Z} 
\caN^{\alpha+1}  \vert Z\vert \frac{t^{\alpha+1}\ep{vt}}{n^{\alpha}},
\end{equation}
with 
$
	\alpha = \begin{cases}\;\;\,m \quad\,, \textrm{ if $m$ is even} \\
								m+1 \,, \textrm{ if $m$ is odd} 
			\end{cases}
$.\\
Here, $\Psi$ is an interaction associated with $O$. The constants $C,v$ depend on $b,b',b_1,\ldots,b_k,k,r$ and the order $m$, but they are independent of $Z,n$ and $t$.
\end{thm}
\begin{proof}
The symmetry for all $m$ is immediate by~(\ref{Higher order construction}) since $\sigma_\mu^{(1)}$ is symmetric. The estimate holds by Theorem~\ref{Thm: anchor} for $m=1,2$, so we proceed by induction. We assume that $\sigma^{(2m)}_\mu$ is a $(2m)$-th order approximation of $\tau_\mu^\Phi$ and that the claim of the theorem holds for $2m$. We write as in \eqref{Telescopic}
\begin{equation}\label{eq:Telescopic-General}
 \tau^\Phi_t(O) - \pi_{t,n}^{(2m+1)}(O)
= \sum _{j = 0}^{n-1} \left(\sigma^{(2m+1)}_\mu\right)^j\left((\tau^\Phi_\mu -\sigma^{(2m+1)}_\mu)\left( (\tau^\Phi_{\mu})^{n-j-1}(O) \right)\right),
\end{equation}
and proceed with an estimate on $(\tau^\Phi_\mu -\sigma^{(2m+1)}_\mu)(\widetilde O)$ for an almost local observable $\widetilde O\in\caL(Y)$. Here again, we decompose the interval $[0,\mu]$ into $r = 2\ell + 1$ intervals according to~(\ref{Higher order construction}) and obtain
\begin{align*}
\tau^\Phi_\mu(\widetilde O) -\sigma_\mu^{(2m+1)}(\widetilde O)
&= \sum_{j=0}^{\ell-1} \left(\sigma_{s_{m}\mu}^{(2m)}\right)^j \circ \left(\tau_{s_{m}\mu}^\Phi - \sigma_{s_{m}\mu}^{(2m)}\right) \circ \left(\tau_{((2\ell-j-1)s_{m} + \tilde s_m)\mu}^{\Phi}\right) (\widetilde O) \\
&\quad + \left(\sigma_{s_{m}\mu}^{(2m)}\right)^\ell \circ \left(\tau^\Phi_{\tilde s_m\mu} - \sigma_{\tilde s_m \mu}^{(2m)}\right) \circ \left(\tau_{\ell s_{m}\mu}^\Phi\right) (\widetilde O)  \\
&\quad + \sum_{j=1}^{\ell} \left(\sigma_{s_{m}\mu}^{(2m)}\right)^{\ell} \circ \sigma_{\tilde s_m\mu}^{(2m)}\circ  \left(\sigma_{s_{m}\mu}^{(2m)}\right)^{j-1} \circ \left(\tau_{s_{m}\mu}^\Phi - \sigma_{s_{m}\mu}^{(2m)}\right) \circ \left(\tau_{(\ell-j)s_{m}\mu}^\Phi\right) (\widetilde O)
\end{align*}
where we denoted $\tilde s_m = 1-(r-1)s_{m}$. By the induction hypothesis, $\left.\frac{d^{j}}{d\mu^{j}}\left(\tau^\Phi_\mu(\widetilde O) - \sigma_\mu^{(2m)}(\widetilde O)\right)\right\vert_{\mu=0}=0$ for all $j=0,\ldots,2m$. This and the above identity imply first of all that the same holds with $\sigma_\mu^{(2m+1)}$ instead of $\sigma_\mu^{(2m)}$, and secondly that
\begin{equation*}
\left.\frac{d^{2m+1}}{d\mu^{2m+1}}\left(\tau^\Phi_\mu(\widetilde O)-\sigma_\mu^{(2m+1)}(\widetilde O)\right)\right\vert_{\mu=0}
= \left(2\ell s_m^{2m+1} + \tilde s_m^{2m+1}\right)\left.\frac{d^{2m+1}}{d\nu^{2m+1}}\left(\tau^\Phi_{\nu}(\widetilde O)-\sigma_{\nu}^{(2m)}(\widetilde O)\right)\right\vert_{\nu=0}.
\end{equation*}
Since $2\ell = r-1$, this vanishes by~(\ref{scalings}), so that $\sigma_\mu^{(2m+1)}$ is a $(2m+1)$-th order approximation. Since $\sigma^{(2m+1)}_\mu$ is time-reversal symmetric we get from Proposition~\ref{prop: TRI} that the $(2m+2)$-th derivative similarly vanishes at $\mu=0$. Hence, 
\begin{equation*}
\tau^\Phi_\mu(\widetilde O)-\sigma_\mu^{(2m+1)}(\widetilde O)
= - \int_{\Sigma^{2m+3}_\mu}\frac{d^{2m+3}}{d{u}^{2m+3}}\left(\sigma_{{u}}^{(2m+1)}\circ\tau^\Phi_{-{u}}\right)\left(\tau^\Phi_\mu(\widetilde O) \right) d^{2m+3}\!{u},
\end{equation*}
and we can proceed as in the proof of~Theorem~\ref{Thm: anchor}. 

There are $r^m(2k-2)+1$ factors in $\sigma_u^{(2m+1)}$ and hence a total of $2^{2m+3}(r^m(k-1)+1)^{2m+3}$ 
terms in the derivative, each of them involving a combination of $2m+3$ derivations from $\{\delta^\Phi\}\cup\{\delta^{j}:j=1,\ldots, k\}$. If ${u}^{(m)}$ is the total time (in absolute value) involved in $\sigma_{{u}}^{(m)}$ (for example, ${u}^{(1)} = {u}^{(2)} = k{u}$), then ${u}^{(2m+1)} = (r-1)s_{m}{u}^{(2m)} + \left\vert (1-(r-1)s_m)\right\vert {u}^{(2m)} = (2(r-1)s_m -1){u}^{(2m)}$ since $1-(r-1)s_m = -s_m(r-1)^{\frac{1}{2m+1}}<0$. Hence,
\begin{equation*}
{u}^{(2m+1)} = \bigg(\prod_{j=1}^m(2(r-1)s_j -1)\bigg) k{u}.
\end{equation*}
Setting $\widetilde O = (\tau^\Phi_{\mu})^{n-j-1}(O)$, we conclude by Propositions~\ref{prop:DynamicsIsAlmostLocalObservable} and~\ref{prop:ActionOfDeltas} that
\begin{multline*}
\left\Inorm\frac{d^{2m+3}}{d{u}^{2m+3}}\left(\sigma_{{u}}^{(2m+1)}\circ\tau^\Phi_{-{u}}\right)\left(\tau^\Phi_{(n-j)\mu}(O) \right)\right\Inorm_{\tilde b, Z} \\
\leq C \Inorm\Psi\Inorm_{b,Z} 
\caN^{2m+3} \ep{c \mu(n-j)} \ep{c ({u}^{(2m +1)} +{u})},
\end{multline*}
for any $\tilde b < \min\{b, 2^{-p}b', 2^{-p}b_1,\ldots, 2^{-p}b_k\}$
where $C,c$ depend on $k$, the rates $b,b',b_1,\ldots,b_k$ as well as the choice of $r$ and $m$. Integrating this over the simplex $\Sigma^{2m+3}_\mu$ and gathering all constants yields
\begin{equation*}
\left\Vert 
\left(\sigma^{(2m+1)}_\mu\right)^j\left((\tau^\Phi_\mu -\sigma^{(2m+1)}_\mu)\left( (\tau^\Phi_{\mu})^{n-j-1}(O) \right)\right)
 \right\Vert
\leq C \Inorm\Psi\Inorm_{b,Z} 
(\caN\mu)^{2m+3}  \vert Z\vert \ep{v t},
\end{equation*}
where the constant $v$ depends again on $k, b,b',b_1,\ldots,b_k, r, m$. 
Since there are $n$ such terms in~(\ref{eq:Telescopic-General}), we have now proved that~(\ref{main formula}) holds for $2m+1$ and therefore also for $2m+2$ by the definition~(\ref{Higher order construction}) of $\sigma_\mu^{(2m+2)}$, concluding the induction.
\end{proof}
\begin{rem}
The theorem should not be misinterpreted as an invitation to take a limit in $m$. Rather, it provides for each fixed $m$ a formula that scales as $n^{-m}$ as $n\to\infty$, while $t$ is arbitrary but fixed. As can be read from the proof, the constant $C$ scales as $\frac{r^{m^2}\ep{r^m}}{m!}$, underlying the importance of picking a small possible $r$, namely $r=3,5$. 
\end{rem}
In the definition~(\ref{Higher order construction}) the interval of size $\mu$ is split into $r = 2\ell+1$ subintervals of width $s_m\mu$ (for $2\ell$ of them) and $(1-(r-1)s_{m})$ (for the middle one). As is clear in the proof (see also~\cite{suzuki1991general}) this choice is largely arbitrary. The claim of the theorem would continue to hold if these $r$ coefficients were replaced by $r$ other real coefficients $\{p_{m,j}:j\in\{1,\ldots,r\}\}$ provided $p_{m,j} = p_{m,r+1-j}$ for all $j=1,\ldots,\ell$ as well as
\begin{equation}\label{pmj Conds}
\sum_{j=1}^r p_{m,j} = 1,\qquad 
\mathrm{and}\qquad\sum_{j = 1}^r p_{m,j}^{2m+1} = 0.
\end{equation}
Reality of the coefficients ensures that $\pi_{t,n}^{(m)}$ are automorphisms and the symmetric choice of coefficient on either side of $p_{m,\ell+1}$ is for time-reversal symmetry. Clearly, there is no non-trivial positive solution of these equations, and (\ref{Higher order construction}) indeed has $p_{m,\ell+1} = -s_m(r-1)^{\frac{1}{2m+1}}<0$, as already pointed out. The appearance of such a negative time evolutions to improve the order of the approximation is reminiscent of the decomposition proposed in~\cite{haah2021quantum}.

\begin{rem}\label{rem:ScalingSmaller1For5}
If $r=3$, then $\vert s_m \vert, \vert 1-(r-1)s_m \vert>1$ with  $\lim_{m\to\infty} s_m =1, \lim_{m\to\infty} (1-(r-1)s_m ) =-1$ and so the individual time intervals in the product scale as $\frac{t}{n}$, independently of $m$ for large $m$. On the other hand, if $r=5,7,\ldots$, then $\vert s_m \vert, \vert 1-(r-1)s_m \vert<1$ with $\lim_{m\to\infty}s_m = \frac{1}{r-2}$ and  $\lim_{m\to\infty} (1-(r-1)s_m ) =-\frac{1}{r-2}$. Hence the individual time intervals in the product scale as $\frac{1}{(r-2)^m}\frac{t}{n}$. The inductive construction and the appearance of negative signs yield a fractal path in the time domain. This behaviour --- already noted in~\cite{suzuki1991general} ---  is exhibited in Figure~\ref{fig:Nice plots}. 

\begin{figure}
\begin{minipage}{0.45\linewidth}
\includegraphics[width=\textwidth]{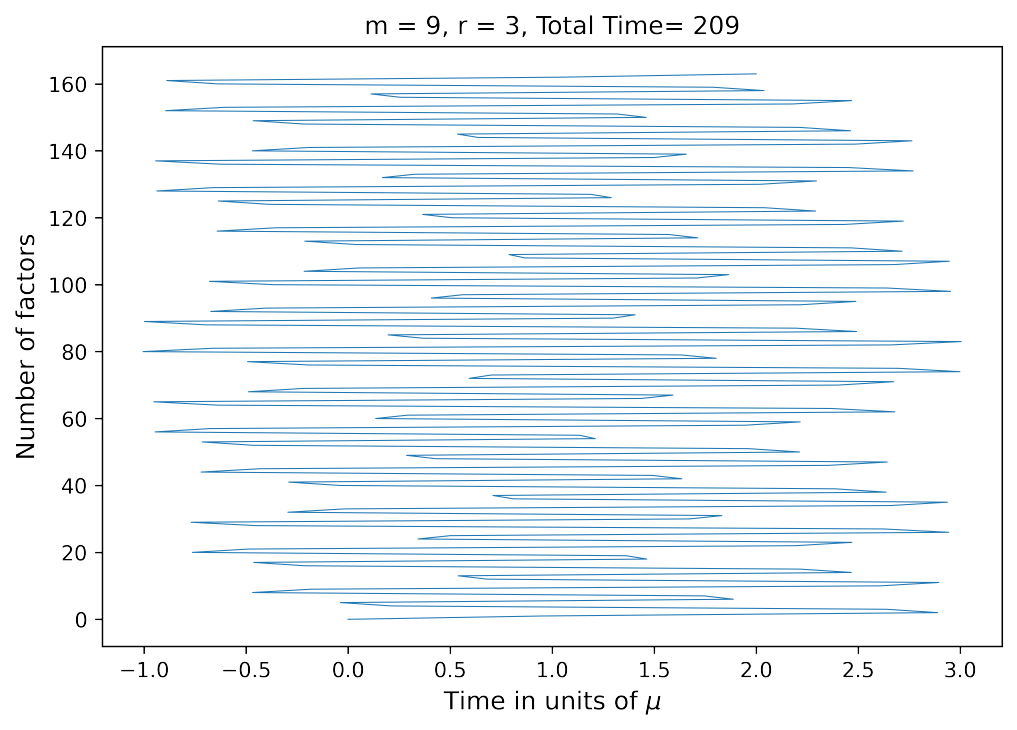}
\end{minipage}\hfill
\begin{minipage}{0.45\linewidth}
\includegraphics[width=\textwidth]{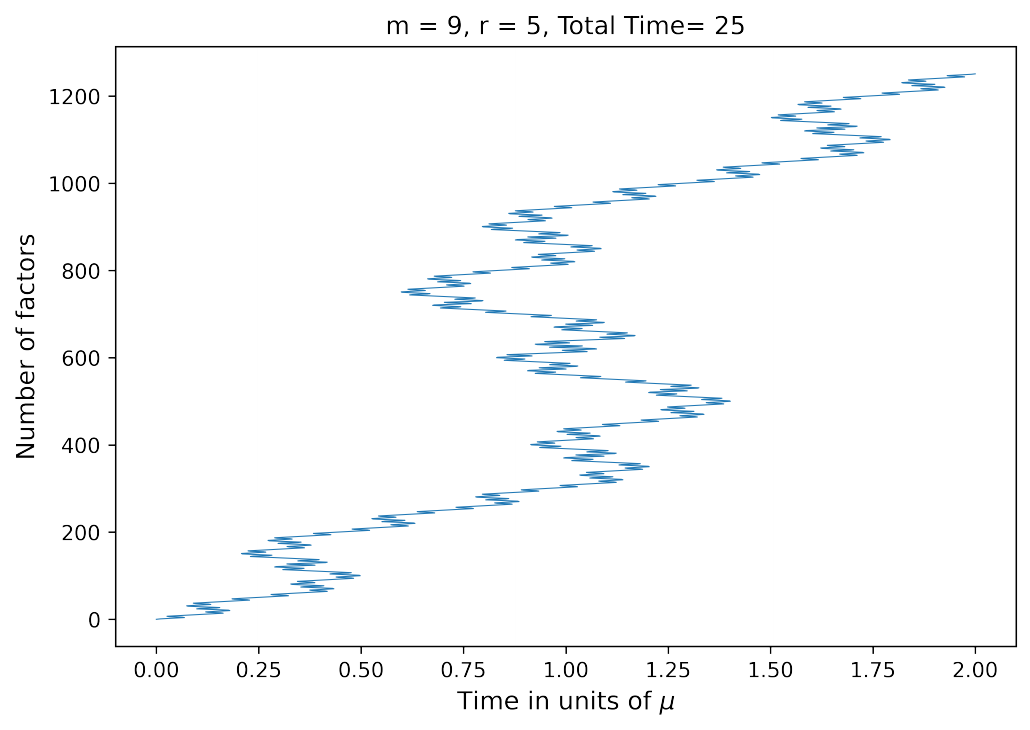}
\end{minipage}
\caption{
The discrete time steps in $\sigma_\mu^{(m)}$ for $r=3$ on the left and $r=5$ on the right, for the same order $m=9$ and in the case $k=2$. There is an order of magnitude difference between the number of terms involved, with the ratio of the number of terms being $\frac{5^42+1}{3^42+1}\simeq 7.7$.}
\label{fig:Nice plots}
\end{figure}
\end{rem}

\section{Quantum simulation: Decomposition in commuting Hamiltonians} \label{sec:QuantumSimulation}

\subsection{Finite depth unitary quantum circuits}  So far, the results are completely general in the sense that they do not require any assumption on the Hamiltonians $K_{j,\Lambda}\in\caL$ beyond their locality. In concrete applications however,
the choice of decomposition of $H$ is determined by the requirement that each $K_{j,\Lambda}$ is a sum of terms acting on spatially disjoint subsets of the lattice and hence mutually commuting. In the simple example of a one-dimensional lattice with nearest-neighbour interaction, namely $\Phi(X) = 0$ if $X\neq\{x,x+1\}$ for some $x\in\bbZ$, one would choose $\Phi_1,\Phi_2$ to be supported on pairs of neighbouring sites $\{2x,2x+1\}$, respectively $\{2x+1,2x+2\}$. 
Each dynamics $\tau_t^{\Phi_1}, \tau_t^{\Phi_2}$ is then strictly local and corresponds to the action of quantum gates in parallel, as illustrated in Figure~\ref{fig:FDQC}.
\begin{figure}
\begin{center}
\includegraphics[width=0.75\textwidth]{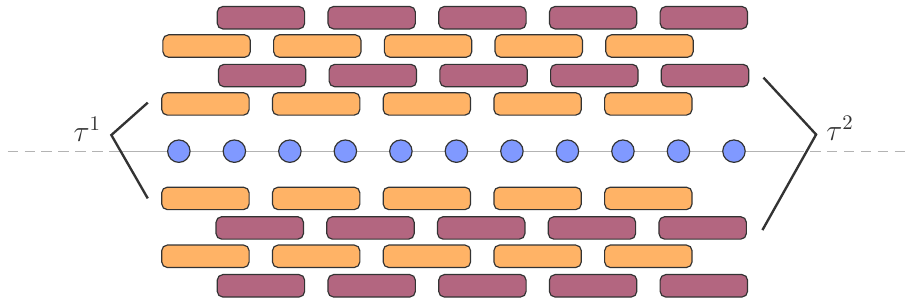}
\end{center}
\caption{A cartoon of the action of a finite depth quantum circuit on a 1-dimensional spin chain, in blue. Since each layer is generated by a commuting Hamiltonian, it is a product of commuting unitaries. Propagation is induced by the alternating action of overlapping layers.}\label{fig:FDQC}
\end{figure}
In this setting where a product formula is referred to as a \emph{finite depth unitary quantum circuit}, Theorem~\ref{Thm: TRI-Trotter} provides a quantitative bound on the error in the approximation of the full dynamics $\tau_t$ by a circuit. The number of factors in the product automorphism is referred to as the \emph{depth} of the circuit. We point out that the following is valid in arbitrary spatial dimensions.
\begin{cor}\label{cor:finiteRange}
Let $\Phi$ be a finite range interaction, namely $\Phi(X) = 0$ if $D(X) >R$ for a fixed $R>0$. Let $Z\in\caP_{\mathrm{fin}}(\Gamma)$ and $O\in\caL(Z)$. For any $m\in\bbN$, there is a finite depth unitary quantum circuit of depth $\caO(\epsilon^{-\frac{1}{m}})$ approximating $\tau_t^\Phi(O)$ within error $\epsilon$, as $\epsilon\to0$.
\end{cor}
\begin{proof}
The depth $h$ of the circuit $\pi_{t,n}^{(m)}$ is proportional to $n$. Hence, imposing that the bound~(\ref{main formula}) is less than $\epsilon$ yields the claim, since $\alpha\geq m$. 
\end{proof}
\noindent With~(\ref{main formula}), we further point out that, as should be expected, the depth of the circuit diverges as $\vert Z\vert^{\frac{1}{m}}$ with the volume of the support $Z$. The depth is furthermore exponential in time, but with a rate $\frac{v}{m}$ that is smaller for a higher order product automorphism. The depth of the circuit proposed in~\cite{haah2021quantum} scales as $\caO\left(t \: \mathrm{polylog}(N t\epsilon^{-1})\right)$, namely better in the total time. However, this is divergent in the size of the system $N$, and uses $\caO\left(\mathrm{polylog}(Nt\epsilon^{-1})\right)$ additional ancilla qubits to achieve the scaling.

In the context of finite range interactions and for a strictly local observable, one may wish to decompose the telescopic sum~(\ref{Telescopic}) rather as
\begin{equation}\label{FDQC:Telescopic}
 \tau^\Phi_t(O) - \pi_{t,n}^{(1)}(O)
= \sum _{j = 0}^{n-1} \left(\tau^\Phi_{\mu}\right)^j\left((\tau^\Phi_\mu -\sigma^{(1)}_\mu)\left( (\sigma^{(1)}_\mu)^{n-j-1}(O) \right)\right),
\end{equation}
since $(\sigma^{(1)}_\mu)^{n-j-1}(O)$ is strictly local. If the $K_{j,\Lambda}$'s are all commuting Hamiltonians, all factors in $(\sigma^{(1)}_\mu)^{n-j-1}$ induce no propagation beyond the range~$R$ of the interactions so that $(\sigma^{(1)}_\mu)^{n-j-1}(O)$ is strictly supported in $Z^{((n-j-1)(2k-1)R)}$ and of operator norm $O$. However, while the action of all derivations $\delta^\Phi,\delta^j$ is again strictly local, increasing the support by $R$, the bound $\Vert \delta^\Upsilon (\widetilde O)\Vert \leq C \Inorm \Upsilon\Inorm_{b,Y} \vert Y\vert \Vert \widetilde O\Vert$ valid for an observable $\widetilde O\in\caA_Y$ yields an estimate
\begin{equation*}
C \vert Z \vert^3 n^{3d} \mu^3\Vert O\Vert,
\end{equation*}
where $d$ is the spatial dimension, for every term of~(\ref{FDQC:Telescopic}). One would therefore obtain $\Vert  \tau^\Phi_t(O) - \pi_{t,n}^{(1)}(O)\Vert = \caO(n^{3d-2})$, emphasizing the need for a careful use of the Lieb-Robinson bound on the very short time intervals of width $\caO(n^{-1})$. This also shows that the \emph{physical propagation is in fact much slower} than what can be read off from the diagram in Figure~\ref{fig:FDQC}. 

\subsection{Long range interactions} While long range interactions pose no issue for our main theorem, Theorem~\ref{Thm: TRI-Trotter}, they cannot be decomposed as a finite depth unitary quantum circuit as just described. A necessary intermediate step is a truncation to finite range interaction. We now comment the error induced by neglecting the interactions between far enough lattice points.

Let $\Phi\in\caB_{b'}$ and $\Phi_R\in\caB_{b'}$ be defined by
\begin{equation*}
\Phi_R(X) = \begin{cases}
\Phi(X) & \text{, if }D(X)\leq R \\ 0 & \text{, otherwise}
\end{cases}.
\end{equation*}
Then, for any $b''<b'$,
\begin{equation*}
\sum_{X\ni x}\frac{\Vert \Phi_R(X) - \Phi(X) \Vert}{\xi_{b''}(D(X))}
= \sum_{X\ni x: D(X)>R}\frac{\Vert \Phi(X) \Vert}{\xi_{b'}(D(X))}\xi_{b'-b''}(D(X))
\end{equation*}
so that $\Inorm \Phi_R - \Phi \Inorm_{b''} \leq \xi_{b'-b''}(R+1) \Inorm \Phi\Inorm_{b'}$. Moreover, 
\begin{equation*}
\tau_t^{\Phi_R}(O) - \tau_t^{\Phi}(O) = \int_0^t \tau_s^{\Phi_R}\left(\delta^{\Phi_R} - \delta^{\Phi}\right)(\tau^{\Phi}_{t-s}(O))ds
\end{equation*}
so that if $O\in\caL_b(Z)$ with associated interaction $\Psi$, 
\begin{align*}
\left \Vert \tau_t^{\Phi_R}(O) - \tau_t^{\Phi}(O)\right\Vert
&\leq C t \sup_{s\in[0,t]}\Vert (\delta^{\Phi_R} - \delta^{\Phi})(\tau^{\Phi}_{s}(O)) \Vert \\
&\leq C t \vert Z\vert \sup_{s\in[0,t]}  \Inorm \Phi_R - \Phi \Inorm_{b''} \Inorm \tau^{\Phi}_{s}(O) \Inorm_{\tilde b,Z} \\
&\leq C t |Z|\ep{\kappa(b',\frac{1}{2}(\tilde{b}+2^{-p}b'))|t|} \Inorm\Phi\Inorm_{b'} \Inorm \Psi \Inorm_{b,Z} \xi_{b'-b''}(R+1)
\end{align*}
by Propositions~\ref{prop:DynamicsIsAlmostLocalObservable},\ref{prop:ActionOfDeltas} and Lemma~\ref{lemma: finite Z interaction}, where $\tilde b<\min\{b,2^{-p}b'\}$, since $\delta^{\Phi_R} - \delta^{\Phi} = \delta^{\Phi_R - \Phi}$. In other words, the error associated with the truncation of the interaction is superpolynomially small in the range $R$. 
In order to achieve an error $\caO(\epsilon)$, the range must be chosen as $R = \caO(\xi_{b'-b''}^{-1}(\epsilon)) = \caO((\log{\epsilon^{-\frac{1}{b'-b''}}})^{\frac{1}{p}})$. 
A commuting decomposition for an interaction of range $R$ requires of the order of $k = R^d$ terms in $d$ dimensions, and since the constant in the Trotter error is exponentially large in the number of terms, (\ref{main formula}) and the lower bound
$\exp\left((\log{\epsilon^{-\frac{1}{b'-b''}}})^{\frac{d}{p}}\right)\geq \epsilon^{-\frac{1}{b'-b''}}$ yield a circuit depth $h = \caO\left(\epsilon^{-\frac{c}{m}}\right)$ 
for a constant $c>1$ that depends on $b'$ and $b''$.

\section*{Acknowledgements}

\noindent Sven Bachmann is supported by NSERC of Canada. S.B.~would like to thank Marcel Schaub for introducing the specific interaction norm used in this work. Markus Lange was supported by NSERC of Canada and also acknowledges financial support from the European Research Council (ERC) under the European Union’s Horizon 2020 research and innovation programme (ERC StG MaMBoQ, grant agreement n.802901).

\appendix
\section{Lieb-Robinson bounds}\label{app:LR-Bounds}
In this section we show that the dynamics generated by an interactions in the class $\caB= \bigcup_{b > 0} \caB_b$ satisfy a Lieb-Robinson bound. The Banach spaces $\caB_b$ of interactions with finite $\Inorm \cdot \Inorm_b$-norm were defined in Section~\ref{sec:Setting}, where we also defined the local Hamiltonian $H\in\caL$ associated with $\Phi\in\caB$.
\begin{prop}
Let $\Phi\in\caB_{b}$ and let $\Lambda\in\caP_{\mathrm{fin}}(\Gamma)$. Let $X,Y\subset\Lambda$ with $X\cap Y = \varnothing$ and $A\in\caA_X,B\in\caA_Y$. For any $b'<b$, we have that
\begin{equation*}
\frac{\Vert [\ep{i t H_\Lambda}A\ep{-i t H_\Lambda},B]\Vert}{\Vert A\Vert \Vert B\Vert}
\leq 
\frac{2 \min\{\vert X\vert,\vert Y \vert\}}{M_{b-b'}}\left(\ep{\kappa(b,b')\vert t\vert}-1\right)\xi_{b'}(d(X,Y))
\end{equation*}
where $\kappa(b,b') = 2\Inorm \Phi\Inorm_{b}M_{b-b'}$ and $M_{\epsilon} := \sup\{\vert X\vert\xi_{\epsilon}(D(X)):X\in\caP_{\mathrm{fin}}(\Gamma)\}$.
\end{prop}
\begin{proof}
We denote $\tau_t^\Lambda(A) = \ep{i t H_\Lambda}A\ep{-i t H_\Lambda}$ and let $f(t) = [\tau_t^\Lambda(A),B]$. Then $f'(t) = \sum_{Z\cap X\neq\varnothing}[i\tau_t^\Lambda([\Phi(Z),A]),B]$ and by Jacobi's identity,
\begin{equation*}
f'(t) = -i[f(t),\sum_{Z\cap X\neq\varnothing}  \tau_t^\Lambda(\Phi(Z))] - \sum_{Z\cap X\neq\varnothing}  i[[B,\tau_t^\Lambda(\Phi(Z))],\tau_t^\Lambda(A)].
\end{equation*}
The first term being norm preserving, we conclude that
\begin{equation*}
\frac{\Vert f(t) \Vert}{\Vert A\Vert} \leq \frac{\Vert [A,B] \Vert}{\Vert A\Vert} 
+ 2  \sum_{Z\cap X\neq\varnothing}\Vert \Phi(Z)\Vert \int_0^{\vert t\vert} \frac{\Vert [\tau_s^\Lambda(\Phi(Z)),B] \Vert}{\Vert \Phi(Z)\Vert}ds,
\end{equation*}
see Lemma~A.1 in~\cite{NacOgaSim06}, namely
\begin{equation*}
C_B(X,t) \leq C_B(X,0) + 2  \sum_{Z\cap X\neq\varnothing}\Vert \Phi(Z)\Vert \int_0^{\vert t\vert} C_B(Z,s)ds
\end{equation*}
where $C_B(X,t) = \sup\{\Vert A\Vert^{-1} \Vert [\tau_t^\Lambda(A),B] \Vert :A\in\caA_X\}$. Iterating this step, it follows that
\begin{equation*}
C_B(X,t) \leq C_B(X,0) +\sum_{n=1}^\infty \frac{2^n\vert t\vert^n}{n!}
\sum_{Z_{n}\cap Z_{n-1}\neq\varnothing}\cdots \sum_{Z_{1}\cap  X\neq\varnothing}  C_B(Z_n,t)\prod_{j=1}^n\Vert \Phi(Z_j)\Vert.
\end{equation*}
Since $C_B(Z,0) = 0$ whenever $Z\cap Y = \varnothing$ and $C_B(Z,0) \leq 2\Vert B\Vert$ otherwise, we conclude that
\begin{equation}\label{Cb bound}
C_B(X,t) \leq 2\Vert B\Vert \delta_{X,Y} + 
2\Vert B\Vert \sum_{n=1}^\infty \frac{2^n\vert t\vert^n}{n!} a_n
\end{equation}
where $\delta_{X,Y} = 0$ if $X\cap Y =\varnothing$ and $\delta_{X,Y} = 1$ otherwise, and we denote $a_n=a_n(X,Y) = \sum_{\substack{Z_{n}\cap Z_{n-1}\neq\varnothing\\Z_n\cap Y\neq\varnothing}}\cdots \sum_{Z_{1}\cap  X\neq\varnothing}  \prod_{j=1}^n\Vert \Phi(Z_j)\Vert$. We claim that
\begin{equation}\label{an}
a_n(X,Y)\leq M_{b-b'}^{n-1}\Inorm \Phi\Inorm_{b}^n\sum_{x\in X}\xi_{b'}(d(x,Y)).
\end{equation}
First of all,
\begin{align*}
 a_1(X,Y) &\leq \sum_{X\cap Z\neq\varnothing,Y\cap Z\neq\varnothing}\Vert\Phi(Z)\Vert
 \leq \sum_{x\in X}\sum_{Z\ni x:Y\cap Z\neq\varnothing}\frac{\Vert\Phi(Z)\Vert}{\xi_{b'}(D(Z))}\xi_{b'}(D(Z))\\
 &\leq \Inorm \Phi\Inorm_{b'}\sum_{x\in X}\xi_{b'}(d(x,Y))
\end{align*}
by monotonicity of $\xi_{b'}$, since $d(x,Y) \leq D(Z)$. The same inequality holds with $X$ and $Y$ exchanged. This is~(\ref{an}) for $n=1$ since $b'<b$ implies that $\Inorm \Phi\Inorm_{b'} \leq \Inorm \Phi\Inorm_{b}$.
We continue by induction, obtaining
\begin{align*}
a_{n+1}(X,Y) 
& = \sum_{Z_{1}\cap  X\neq\varnothing} \Vert \Phi(Z_1)\Vert a_n(Z_1,Y)\\
&\leq M_{b-b'}^{n-1}\Inorm \Phi\Inorm_{b}^n\sum_{x\in X}\sum_{Z_1\ni x}\frac{\Vert\Phi(Z_1)\Vert}{\xi_{b}(D(Z_1))}\xi_{b}(D(Z_1))\sum_{z\in Z_1}\xi_{b'}(d(z,Y))
\end{align*}
We factorize $\xi_b(r) = \xi_{b-b'}(r)\xi_{b'}(r)$ and bound
\begin{equation*}
\xi_{b'}(D(Z_1)) \xi_{b'}(d(z,Y))
\leq \xi_{b'}(d(x,z) + d(z,Y))
\leq \xi_{b'}(d(x,Y))
\end{equation*}
by monotonicity and logarithmic subadditivity, since $D(Z_1)\geq d(x,z)$. With this,
\begin{equation*}
a_{n+1}(X,Y) 
\leq M_{b-b'}^{n}\Inorm \Phi\Inorm_{b}^{n+1}\sum_{x\in X}\xi_{b'}(d(x,Y))
\end{equation*}
as announced.

The sets $X,Y$ appearing symmetrically in the estimates above, the right hand side of~(\ref{an}) can be improved to the minimum of $\sum_{x\in X}\xi_{b'}(d(x,Y))$ and $\sum_{y\in Y}\xi_{b'}(d(y,X))$. It remains to use
\begin{equation*}
\min\Big\{\sum_{x\in X}\xi_{b'}(d(x,Y)),\sum_{y\in Y}\xi_{b'}(d(y,X)) \Big\}\leq \min\{\vert X\vert,\vert Y\vert\} \xi_{b'}(d(X,Y))
\end{equation*}
and to plug the resulting bound into~(\ref{Cb bound}) to get
\begin{equation*}
C_B(X,t) \leq 2\Vert B\Vert \left[\delta_{X,Y} + \frac{\min\{\vert X\vert,\vert Y\vert\} \xi_{b'}(d(X,Y))}{M_{b-b'}} \sum_{n=1}^\infty \frac{(2\vert t\vert\Inorm \Phi\Inorm_{b}M_{b-b'})^n}{n!}\right].
\end{equation*}
If $X\cap Y = \varnothing$, this reads
\begin{equation*}
\Vert [\tau_t^\Lambda(A),B] \Vert \leq  \frac{2\Vert A\Vert\Vert B\Vert \min\{\vert X\vert,\vert Y\vert\} }{M_{b-b'}} \left(\ep{2\vert t\vert\Inorm \Phi\Inorm_{b}M_{b-b'}}-1\right) \xi_{b'}(d(X,Y))
\end{equation*}
which is the claim of the proposition.
\end{proof}
\noindent The proof yields the following bound that is valid for any $X,Y$ not necessarily disjoint:
\begin{multline*}
\frac{\Vert [\ep{i t H_\Lambda}A\ep{-i t H_\Lambda},B]\Vert}{\Vert A\Vert \Vert B\Vert}\\
\leq \frac{2 \Vert A\Vert\Vert B\Vert}{M_{b-b'}}
g_{b,b'}(t)
\min\Big\{\sum_{x\in X}\xi_{b'}(d(x,Y)),\sum_{y\in Y}\xi_{b'}(d(y,X)) \Big\}
\end{multline*}
where $g_{b,b'}(t) = \ep{\kappa(b,b')\vert t\vert}-(1-\delta_{X,Y}$).

As pointed out earlier, the proof runs along the general lines of~\cite{NacOgaSim06}. It only differs in the estimate of $a_n$ because of the choice of a different norm and a slightly more general class of interactions. In particular, in the case $p<1$, the subexponential decay in $d(X,Y)$ has its origin in the subexponential decay of the interaction~$\Phi\in\caB_b$.

It is also a well-known fact that the Lieb-Robinson bound yields the existence of the dynamics in the infinite volume limit. We provide here a short proof in the specific setting of this paper. We now consider an increasing sequence of subsets $\Lambda_n$ that is absorbing in the sense that for any $x\in\Gamma$, there is $N$ such that $x\in\Lambda_n$ for all $n\geq N$. 

\begin{cor}
Let $\Phi\in\caB_{b}$, let $X\in\caP_{\mathrm{fin}}(\Gamma)$ and $A\in\caA_X$. For all $n$ such that $X\subset \Lambda_n$, let $\tau^n_t(A) = \ep{i t H_{\Lambda_n}}A\ep{-i t H_{\Lambda_n}}$. The sequence $\tau^n_t(A)$ is convergent to $\tau_t(A)$. Moreover, $t\mapsto\tau_t$ extends to a strongly continuous group of automorphisms on $\caA$. 
\end{cor}
\begin{proof}
We note that if $n>m$, then
\begin{equation*}
\tau^m_{-t}\circ \tau^n_t(A) - A 
= \int_0^t \tau^m_{-s}\left(i[H_{\Lambda_n}- H_{\Lambda_m},\tau^n_s(A)]\right)ds
\end{equation*}
Since $H_{\Lambda_n}- H_{\Lambda_m} = \sum_{Z\cap(\Lambda_n\setminus\Lambda_m)\neq\varnothing}\Phi(Z)$, we conclude that
\begin{equation*}
\Vert \tau^n_t(A) - \tau^m_t(A)\Vert
\leq \sum_{Z\cap(\Lambda_n\setminus\Lambda_m)\neq\varnothing} \int_0^{\vert t\vert} \Vert [\Phi(Z) , \tau^n_s(A)]\Vert ds
\end{equation*}
The Lieb-Robinson bound now yields
\begin{equation*}
\sum_{Z\cap(\Lambda_n\setminus\Lambda_m)\neq\varnothing} \Vert [\Phi(Z) , \tau^n_s(A)]\Vert
\leq \frac{2\Vert A\Vert\ep{\kappa(b,b')\vert s\vert}}{M_{b-b'}} \sum_{Z\cap(\Lambda_n\setminus\Lambda_m)\neq\varnothing}  \Vert \Phi(Z)\Vert \sum_{x\in X}\xi_{b'}(d(x,Z))
\end{equation*}
for any $b'<b$. The sum over $Z$ can be upper bounded by $\sum_{z\in\Lambda_n\setminus\Lambda_m}\sum_{Z\ni z}$. We introduce the factor $\xi_{b'}(D(Z))$ and use the logarithmic superadditivity of $\xi_{b'}$ and finally the bound $d(x,z) \leq d(x,Z) + D(Z)$ to get
\begin{multline*}
\sum_{Z\cap(\Lambda_n\setminus\Lambda_m)\neq\varnothing}  \Vert \Phi(Z)\Vert \sum_{x\in X}\xi_{b'}(d(x,Z)) \\
\leq \sum_{z\in\Lambda_n\setminus\Lambda_m}\sum_{Z\ni z} \frac{\Vert \Phi(Z)\Vert}{\xi_{b'}(D(Z))}
\sum_{x\in X}\xi_{b'}(d(x,Z) + D(Z)) 
\end{multline*}
Note that $\Phi\in\caB_b$ implies that $\Inorm \Phi\Inorm_{b'}<\infty$ for all $b'<b$. It remains to use the bound $d(x,z) \leq d(x,Z) + D(Z)$ to get
\begin{equation*}
\Vert \tau^n_t(A) - \tau^m_t(A)\Vert
\leq \frac{2\Vert A\Vert\Inorm \Phi\Inorm_{b'}} {M_{b-b'}\kappa(b,b')} \ep{\kappa(b,b')|t|}\sum_{z\in\Lambda_n\setminus\Lambda_m}\sum_{x\in X}\xi_{b'}(d(x,z)).
\end{equation*}
The summability of $\xi_{b'}$ and the fact that $X$ is a finite set implies that the sum vanishes as $n,m\to\infty$. In other words, $(\tau^n_t(A))_{n}$ is Cauchy sequence in $\caA$ (uniformly in $t$ for $t$ in a compact set) and hence convergent. 

The limiting map $\tau_t$ is bounded (since $\Vert \tau_t(A)\Vert = \Vert A\Vert $) on the dense set of local observables. Therefore, it extends to a bounded linear map on $\caA$. The group property follows from that of the finite volume approximations.
\end{proof}


\begin{thebibliography}{10}

\bibitem{trotter1959product}
H.F. Trotter.
\newblock On the product of semi-groups of operators.
\newblock {\em Proc. Amer. Math Soc.}, 10(4):545--551, 1959.

\bibitem{chernoff1968note}
P.R. Chernoff.
\newblock Note on product formulas for operator semigroups.
\newblock {\em J. Func. Ana.}, 2(2):238--242, 1968.

\bibitem{BratRob1}
O.~Bratteli and D.W. Robinson.
\newblock {\em Operator Algebras and Quantum Statistical Mechanics: Volume 1:
  C*-and W*-Algebras. Symmetry Groups. Decomposition of States}.
\newblock Springer Science \& Business Media, 2012.

\bibitem{kato1978trotter}
T.~Kato.
\newblock Trotter's product formula for an arbitrary pair of self-adjoint
  contraction semigroup.
\newblock {\em Topics in Func. Anal., Adv. Math. Suppl. Studies}, 3:185--195,
  1978.

\bibitem{ichinose1980product}
T.~Ichinose.
\newblock A product formula and its application to the {Schr{\"o}dinger}
  equation.
\newblock {\em Publ. Res. Inst. Math. Sci.}, 16(2):585--600, 1980.

\bibitem{ReedSimonI}
M.~Reed and B.~Simon.
\newblock {\em Methods of Modern Mathematical Physics. I: Functional Analysis.}
\newblock Academic Press, 1980.

\bibitem{nelson1964feynman}
E.~Nelson.
\newblock {Feynman integrals and the Schr{\"o}dinger equation}.
\newblock {\em J. Math. Phys.}, 5(3):332--343, 1964.

\bibitem{exner2005product}
P.~Exner and T.~Ichinose.
\newblock A product formula related to quantum zeno dynamics.
\newblock {\em Annales H. Poincar\'e}, 6(2):195--215, 2005.

\bibitem{zagrebnov2019gibbs}
V.A. Zagrebnov.
\newblock {\em Gibbs semigroups}, volume 273 of {\em Operator Theory: Advances
  and Applications}.
\newblock Springer, 2019.

\bibitem{Llo96}
S.~Lloyd.
\newblock Universal quantum simulators.
\newblock {\em Science}, 273(5278):1073--1078, 1996.

\bibitem{BerAhoCleSan06}
D.W. Berry, G.~Ahokas, R.~Cleve, and B.C. Sanders.
\newblock {Efficient Quantum Algorithms for Simulating Sparse Hamiltonians}.
\newblock {\em Commun. Math. Phys.}, 270(2):359--371, 2006.

\bibitem{wiebe2011simulating}
N.~Wiebe, D.W. Berry, P.~H{\o}yer, and B.C. Sanders.
\newblock Simulating quantum dynamics on a quantum computer.
\newblock {\em J. Phys. A}, 44(44):445308, 2011.

\bibitem{chen2010local}
X.~Chen, Z.-Ch. Gu, and X.-G. Wen.
\newblock Local unitary transformation, long-range quantum entanglement, wave
  function renormalization, and topological order.
\newblock {\em Phys. Rev. B}, 82(15):155138, 2010.

\bibitem{NatureSimulationExp}
H.~Bernien, S.~Schwartz, A.~Keesling, H.~Levine, A.~Omran, H.~Pichler, S.~Choi,
  A.S. Zibrov, M.~Endres, M.~Greiner, V.~Vuleti{\'c}, and M.D. Lukin.
\newblock Probing many-body dynamics on a 51-atom quantum simulator.
\newblock {\em Nature}, 551(7682):579--584, 2017.

\bibitem{SmiKimPolKno19}
A.~Smith, M.~Kim, F.~Pollmann, and J.~Knolle.
\newblock Simulating quantum many-body dynamics on a current digital quantum
  computer.
\newblock {\em npj Quantum Inf.}, 5, 12 2019.

\bibitem{suzuki1991general}
M.~Suzuki.
\newblock General theory of fractal path integrals with applications to
  many-body theories and statistical physics.
\newblock {\em J. Math. Phys.}, 32(2):400--407, 1991.

\bibitem{hatano2005finding}
N.~Hatano and M.~Suzuki.
\newblock Finding exponential product formulas of higher orders.
\newblock In {\em Quantum annealing and other optimization methods}, pages
  37--68. Springer, 2005.

\bibitem{PhysRevX.11.011020}
A.M. Childs, Y.~Su, M.C. Tran, N.~Wiebe, and Sh. Zhu.
\newblock {Theory of Trotter Error with Commutator Scaling}.
\newblock {\em Phys. Rev. X}, 11:011020, Feb 2021.

\bibitem{monaco2019adiabatic}
D.~Monaco and S.~Teufel.
\newblock Adiabatic currents for interacting fermions on a lattice.
\newblock {\em Rev. Math. Phys.}, 31(03):1950009, 2019.

\bibitem{sahinoglu2021hamiltonian}
B.~Sahinoglu and R.D. Somma.
\newblock Hamiltonian simulation in the low energy subspace.
\newblock {\em Bull. Amer. Phys. Soc.}, 2021.

\bibitem{haah2021quantum}
J.~Haah, M.B. Hastings, R.~Kothari, and G.H. Low.
\newblock Quantum algorithm for simulating real time evolution of lattice
  hamiltonians.
\newblock {\em SIAM J. Comp.}, Special Section FOCS 2018:350--360, 2021.

\bibitem{bachmann2017local}
S.~Bachmann and A.~Bluhm.
\newblock Local factorisation of the dynamics of quantum spin systems.
\newblock {\em J. Math. Phys.}, 58(7):071901, 2017.

\bibitem{ALPUs}
D.~Ranard, M.~Walter, and F.~Witteveen.
\newblock {A converse to Lieb-Robinson bounds in one dimension using index
  theory}.
\newblock {\em arXiv preprint arXiv:2012.00741}, 2020.

\bibitem{GNVW}
D.~Gross, V.~Nesme, H.~Vogts, and R.F. Werner.
\newblock Index theory of one dimensional quantum walks and cellular automata.
\newblock {\em Commun. Math. Phys.}, 310(2):419--454, 2012.

\bibitem{BratRob2}
O.~Bratteli and D.W. Robinson.
\newblock {\em Operator Algebras and Quantum Statistical Mechanics: Volume 2:
  Equilibrium states. Models in Quantum Statitstical Mechanics}.
\newblock Springer Science \& Business Media, 2012.

\bibitem{NacOgaSim06}
B.~Nachtergaele, Y.~Ogata, and R.~Sims.
\newblock Propagation of correlations in quantum lattice systems.
\newblock {\em J. Stat. Phys.}, 124(1):1--13, 2006.

\bibitem{NacSimYou19-1}
B.~Nachtergaele, R.~Sims, and A.~Young.
\newblock {Quasi-locality bounds for quantum lattice systems. I. Lieb-Robinson
  bounds, quasi-local maps, and spectral flow automorphisms}.
\newblock {\em J. Math. Phys.}, 60(6):061101, 2019.

\bibitem{lieb1972finite}
E.H. Lieb and D.W. Robinson.
\newblock The finite group velocity of quantum spin systems.
\newblock {\em Commun. Math. Phys.}, 28:251--257, 1972.

\end{thebibliography}
\end{document}